\documentclass[pdflatex,sn-mathphys-num]{sn-jnl}

\usepackage{graphicx}%
\usepackage{multirow}%
\usepackage{amsmath,amssymb,amsfonts}%
\usepackage{amsthm}%
\usepackage{mathrsfs}%
\usepackage[title]{appendix}%
\usepackage{xcolor}%
\usepackage{textcomp}%
\usepackage{manyfoot}%
\usepackage{booktabs}%
\usepackage{algorithm}%
\usepackage{algorithmicx}%
\usepackage{algpseudocode}%
\usepackage{listings}%

\newtheorem{theorem}{Theorem}
\newtheorem{assumption}{Assumption}
\newtheorem{corollary}{Corollary}
\newtheorem{remark}{Remark}
\newtheorem{lemma}{Lemma}
\newtheorem{proposition}{Proposition}

\newcommand{\pd}[2]{\frac{\partial #1}{\partial #2}}
\newcommand{\dd}{\mathrm{d}}
\newcommand{\ii}{\mathrm{i}}
\newcommand{\ee}{\mathrm{e}}
\newcommand{\Rey}{R}
\newcommand{\Ai}{\mathrm{Ai}}
\newcommand{\Imag}{\operatorname{Im}}
\newcommand{\Real}{\operatorname{Re}}

\theoremstyle{thmstyleone}%

\theoremstyle{thmstyletwo}%

\theoremstyle{thmstylethree}%

\begin{document}

\title{\bf Spectral perturbation theory for wall-admittance effects on compressible boundary-layer instability}

\author[1]{\fnm{Jiguang} \sur{Yu}}\email{jyu678@bu.edu}
\equalcont{This author contributed equally to this work as co-first authors.}

\author*[2]{\fnm{Louis Shuo} \sur{Wang}}\email{swang116@vols.utk.edu}
\equalcont{This author contributed equally to this work as co-first authors.}

\author*[3]{\fnm{Ye} \sur{Liang}}\email{zcahiad@ucl.ac.uk}
\equalcont{This author contributed equally to this work as co-first authors.}

\affil[1]{\orgdiv{College of Engineering},
  \orgname{Boston University},
  \orgaddress{\city{Boston}, \postcode{02215}, \state{MA}, \country{United States}}}

\affil[2]{\orgdiv{Department of Mathematics}, 
  \orgname{Northeastern University}, 
  \orgaddress{\city{Boston}, \postcode{02115}, \state{MA}, \country{United States}}}

\affil[3]{\orgdiv{Department of Mathematics}, 
  \orgname{University College London}, 
  \orgaddress{\city{London}, \postcode{WC1E 6BT}, \country{United Kingdom}}}

\abstract{Thin wall treatments modify high-speed boundary-layer instability through
the pressure they admit or absorb at the wall.  This paper develops a
unified admittance formulation for such effects on trapped compressible
Rayleigh modes.  For a simple rigid-wall eigenpair, we prove the spectral
sensitivity law
\[
        c(A)=c_0+KA+\mathcal O(|A|^2),
        \qquad
        \delta\sigma=\alpha\Imag(KA)+\mathcal O(|A|^2),
\]
where \(A\) is the wall admittance and \(K\) is an explicit functional of
the rigid-wall eigenfunction.  The formula separates wall physics from
outer-mode physics and yields a phase criterion for stabilisation.  Matched
asymptotics show that viscous and thermal wall layers, blind-pore coatings
and shallow non-separating roughness all reduce to this same boundary
condition, with additive leading admittances.  Mach-4.5 computations
validate the sensitivity coefficient and demonstrate porous damping,
viscous-wall damping and sign-changing reactive roughness effects.}



\maketitle

\section{Introduction and main results}
\label{sec:intro}

The growth of instability waves in laminar high-speed boundary layers
controls laminar--turbulent transition and thereby surface heating and drag
on supersonic and hypersonic vehicles \citep{mack1984boundary,fedorov2011transition,wang2025analysis,schuabb2025hypersonic,zhuang2024instability,chen2022boundary,wu2023new,cheng2024progress,wang2025analysis1,cao2023stability,gao2022rolling}.  
At the Mach numbers of interest the dominant instabilities are the first
(inflectional) mode, rendered inviscidly unstable by the generalised
inflection point \citep{yu2026from,lees1946investigation,liang2023inviscid,liu2025bidirectional,marxen2014direct,dong2021asymptotic}, and the second (Mack) mode, a trapped
acoustic mode of the waveguide formed between the wall and the sonic line
\citep{mack1984boundary,dong2020receptivity,lu2025acoustic}.  Because these modes are inviscid at leading
order, with $\mathcal O(1)$ pressure fluctuations at the wall, they are sensitive to
what the wall does with that pressure: to the transpiration admitted by a
porous coating \citep{cai2026optimal,fedorov2001stabilization,liang2025global,rasheed2002experiments,fedorov2011transition}, to the
displacement of the effective impermeable plane by distributed roughness
\citep{fujii2006experiment,fransson2006delaying,chuvakhov2016spontaneous,wang2026algebraic}, and to the mass flux pumped by
the thin viscous Stokes and thermal layers that survive at the wall even
when the mode itself is inviscid \citep{mack1984boundary,wang2025multi,gaponov1975effect}.

Each of these wall effects possesses a substantial literature and its own
formalism: viscous corrections are computed from full Orr--Sommerfeld-type
systems or from triple-deck theory \citep{smith1979non,yu2026beyond,skorokhodov2007numerical,gallagher2016effect,wang2026damage,landahl1990sublayer,huang1983application,rizzetta1978triple}; coatings are inserted
as semi-empirical acoustic impedances \citep{fedorov2001stabilization,bres2013second,
tritarelli2015stabilization}; roughness is treated by homogenisation
\citep{bechert1989viscous,yu2026pattern,luchini1991resistance} coupled to receptivity or stability
analyses.  The novelty of the present paper is not the separate existence of
these reductions.  \emph{The contribution is the asymptotic and spectral
synthesis}: under explicit ordering assumptions, all three mechanisms ---
and their combinations --- reduce to the same admittance boundary condition
on the same outer compressible Rayleigh problem, and the leading effect of
any of them on a trapped eigenmode is determined by one sensitivity
coefficient computable from the rigid-wall eigenfunction.  This separates
wall physics, represented by a complex admittance $A$, from outer-mode
physics, represented by a complex coefficient $K$, and gives an explicit
phase criterion for stabilisation.  The admittance description itself has a
long history for compliant walls in incompressible flow
\citep{benjamin1960effects,landahl1962stability,wang2026breakdown,carpenter1985hydrodynamic,koseki2025understanding,malik2018growth} and for liner acoustics
\citep{zwikker1949sound,yu2026rigorous,melling1973acoustic}; its use here as the common endpoint
of matched asymptotics for viscous, porous and rough walls, together with
the sensitivity calculus, is, to our knowledge, new.

The main results are as follows.  Let $\hat p(y)\ee^{\ii\alpha(x-ct)}$ be a
pressure disturbance of the compressible Rayleigh problem
(\S\ref{sec:rayleigh}) on a boundary layer with profiles $U(y),T(y)$ at Mach
number $M$, and let the wall enter only through the admittance
\begin{equation}
A(\alpha,c)=-\,\frac{\hat v(0^+)}{\hat p(0^+)},
\label{eq:Adef-intro}
\end{equation}
the ratio of outward wall-normal velocity to pressure seen by the outer flow
at the wall ($A=0$: rigid; $\Real A\ge0$: passive).  Then:
\begin{enumerate}
\item[(i)] (\emph{Theorem~\ref{thm:sens}, \S\ref{sec:theorem}.})  If
$(\hat p_0,c_0)$ is a simple, trapped, rigid-wall eigenpair, the perturbed
eigenvalue is
\begin{equation}
c(A)=c_0+KA+\mathcal O(|A|^2),
\qquad
K=\frac{\ii\alpha\,\hat p_0(0)^2}{c_0\,I},
\label{eq:mainresult}
\end{equation}
with $I$ the explicit quadratic functional \eqref{eq:Idef} of the rigid-wall
eigenfunction; consequently the growth rate $\sigma=\alpha c_i$ shifts by
$\delta\sigma=\alpha\Imag(KA)+\mathcal O(|A|^2)$, and an admittance of phase
$\varphi$ stabilises the mode if and only if
$\Imag(K\ee^{\ii\varphi})<0$ (Corollary~\ref{cor:phase}).
\item[(ii)] (\emph{\S\ref{sec:mechanisms}.})  Under the thin-layer ordering
of Assumption~\ref{ass:thin}, the viscous Stokes and thermal sublayers, a
coating of blind cylindrical pores, and shallow, dense, non-separating
roughness possess the closed-form admittances \eqref{eq:Av},
\eqref{eq:AT}, \eqref{eq:Ap} and \eqref{eq:Ar}, additive in combination.
\item[(iii)] (\emph{\S\ref{sec:numerics}.})  For a Mach-4.5 adiabatic
flat-plate boundary layer, direct solutions of the admittance eigenvalue
problem validate \eqref{eq:mainresult} to relative errors below
$5\times10^{-4}$ (table~\ref{tab:val}) and confirm the phase criterion in
detail; in particular the sign of the purely reactive roughness effect
reverses along the second-mode branch exactly where $\arg K$ crosses
$-\pi/2$.
\end{enumerate}

The objective is deliberately limited.  We do not seek a replacement for
full viscous stability theory, direct simulation, or spatial transition
prediction.  We identify the common asymptotic structure of several thin
wall mechanisms, prove the spectral perturbation result that makes the
reduction useful, and demonstrate both at a representative hypersonic
condition.  The Mach-4.5 calculations are a proof of concept: they validate
the sensitivity formula against direct eigenvalue computations and show how
porous absorption, viscous wall layers and shallow roughness occupy
different regions of the same complex-admittance plane.  Amplitude-dependent
and acoustic admittances can be inserted into the same boundary condition;
these extensions are outlined, deliberately briefly, in Appendix~\ref{appsec:extensions}.

The paper is organised as follows.  Section~\ref{sec:rayleigh} formulates
the compressible Rayleigh problem with wall admittance and records the
passivity constraint.  Section~\ref{sec:theorem} states the spectral
setting, proves Theorem~\ref{thm:sens} and its phase corollary, and
delimits their use.  Section~\ref{sec:mechanisms} derives the effective
admittances of the thin wall mechanisms under explicit ordering
assumptions.  Section~\ref{sec:numerics} validates and demonstrates the
theory at Mach 4.5.  Section~\ref{sec:discussion} discusses limitations and
extensions.  Appendices contain the uniformly valid (Airy-function) form of
the viscous-layer admittance and its triple-deck connection, the pore-scale
derivation of the coating coefficients, the far-field tail contribution to
the sensitivity functional, and numerical details.

\section{The compressible Rayleigh problem with wall admittance}
\label{sec:rayleigh}

Velocities are scaled by the free-stream speed $U_\infty^*$, temperature and
density by their free-stream values, pressure perturbations by
$\rho_\infty^*U_\infty^{*2}$, and lengths by the local Blasius scale
$L^*=(2\nu_\infty^*x^*/U_\infty^*)^{1/2}$, so that
$\Rey=U_\infty^*L^*/\nu_\infty^*=(2Re_x)^{1/2}\gg1$.  The base flow is a
compressible boundary layer $U(y),T(y)$ of a perfect gas
($\gamma=1.4$); the numerical illustrations use the compressible Blasius
solution with the Chapman law $\mu=T$, unit base-flow Prandtl number and an
adiabatic wall at $M=4.5$, for which $T_w=5.05$,
$\lambda_w\equiv U'(0)=0.0930$ and $\delta_{99}=10.27$
(figure~\ref{fig:1}a).  Disturbances take the locally parallel form
$(\hat u,\hat v,\hat p,\hat\theta)(y)\,\ee^{\ii\alpha(x-ct)}+\mathrm{c.c.}$
with $\alpha$ real and $c=c_r+\ii c_i$ complex; the temporal growth rate is
$\sigma=\alpha c_i$.

Outside all wall layers the pressure obeys the compressible Rayleigh
equation
\begin{equation}
\hat p''-\left(\frac{2U'}{U-c}-\frac{T'}{T}\right)\hat p'
-\alpha^2\!\left(1-\frac{M^2(U-c)^2}{T}\right)\hat p=0,
\label{eq:rayleigh}
\end{equation}
with the wall-normal velocity recovered from the $y$-momentum equation
$\ii\alpha\bar\rho(U-c)\hat v=-\hat p'$, $\bar\rho=1/T$.  Equation
\eqref{eq:rayleigh} admits the self-adjoint form
\[
\big(\Phi\,\hat p'\big)'=\alpha^2\big(\Phi-M^2\big)\hat p,
\qquad
\Phi(y;c)=\frac{T}{(U-c)^2}.
\]
In the uniform stream the decaying solution is
\begin{equation}
\hat p\sim\exp[-\gamma\,(y-y_b)],
\qquad
\gamma=\alpha\big[1-M^2(1-c)^2\big]^{1/2},\quad\Real\gamma>0,
\label{eq:farfield}
\end{equation}
applied at the edge $y_b$ of the computational domain; modes with
$\Real\gamma>0$ are trapped.

All wall physics is condensed into the admittance \eqref{eq:Adef-intro},
where $0^+$ denotes the bottom of the outer region after the wall sublayers
have been absorbed.  Combining \eqref{eq:Adef-intro} with the $y$-momentum
equation at $y=0$ (where $U=0$) gives the outer wall condition in pressure
form,
\begin{equation}
\hat p'(0)=-\,\frac{\ii\alpha c\,A}{T_w}\,\hat p(0).
\label{eq:wallBC}
\end{equation}
Here positive $\hat v$ is taken in the outward normal direction, from the
wall into the fluid, so that $-\overline{p'v'}\,|_{y=0}$ is the acoustic
power per unit area entering the wall; for a neutral wave it equals
$\tfrac12|\hat p_w|^2\Real A$, whence a \emph{passive} wall --- one that can
only absorb disturbance energy --- has $\Real A\ \ge\ 0$, with the standard generalisation for $c_i\neq0$.  Passivity serves
repeatedly below as a sign check on the asymptotics; it does \emph{not} by
itself imply stabilisation, which is governed by the phase of $A$ relative
to that of the mode's sensitivity coefficient (Corollary~\ref{cor:phase}).
We present two extensions in Appendix~\ref{appsec:extensions} to illustrate the reach of the boundary condition
\eqref{eq:wallBC} beyond linear, trapped dynamics; both are presented as
outlines, not as developed theories.

\section{The spectral sensitivity theorem}
\label{sec:theorem}

\subsection{Dispersion function and analytic eigenvalue branch}
\label{ssec:dispersion}

We first make explicit the analytic eigenvalue branch to which the
sensitivity calculation applies.  Fix \(\alpha\) and choose the branch of
\[
        \gamma(c)=\alpha\big[1-M^2(1-c)^2\big]^{1/2}
\]
with \(\Real\gamma(c_0)>0\).  In a neighbourhood of a trapped eigenvalue
\(c_0\), away from the sonic branch points and from critical singularities
on the real \(y\)-axis, the coefficients of the Rayleigh equation
\eqref{eq:rayleigh} and the far-field exponent \(\gamma(c)\) are analytic
functions of \(c\).  Let \(P(y;c)\) denote the solution of
\eqref{eq:rayleigh} obtained by shooting from \(y=y_b\) with the normalisation
\[
        P(y_b;c)=1,\qquad
        P'(y_b;c)=-\gamma(c).
\]
Standard analytic dependence of solutions of linear ordinary differential
equations on parameters implies that \(P(0;c)\) and \(P'(0;c)\) are analytic
functions of \(c\) in this neighbourhood.

For a wall admittance \(A\), define the scalar dispersion function
\[
        D(c,A)
        =
        P'(0;c)
        +\frac{\ii\alpha c A}{T_w}P(0;c).
\]
The eigenvalues of the Rayleigh problem with admittance boundary condition
\eqref{eq:wallBC} are precisely the zeros of \(D(c,A)\).  In particular,
the rigid-wall eigenvalue satisfies
$D(c_0,0)=P'(0;c_0)=0$.
We call \(c_0\) simple if
\begin{equation}
        D_c(c_0,0)\ne0 .
\label{eq:simple-dispersion}
\end{equation}
Equivalently, the zero of the rigid-wall dispersion relation
\(D(c,0)=P'(0;c)\) at \(c=c_0\) is simple.  Under this condition, the
implicit function theorem gives a unique analytic eigenvalue branch
\(c=c(A)\), for \(|A|\) sufficiently small, such that
$D(c(A),A)=0$ and $c(0)=c_0$.
Differentiating \(D(c(A),A)=0\) at \(A=0\) gives
$\displaystyle c'(0)=-\frac{D_A(c_0,0)}{D_c(c_0,0)}$.
The remainder of the proof of Theorem~\ref{thm:sens} evaluates this quotient
by Green's identity and shows that it is equal to
$\displaystyle \frac{\ii\alpha\hat p_0(0)^2}{c_0 I}$.

\begin{assumption}[Spectral setting]
\label{ass:spectral}
The rigid-wall eigenvalue \(c_0\) is a trapped eigenvalue of the dispersion
relation \(D(c,0)=0\), with \(\Real\gamma(c_0)>0\), and is simple in the
sense of \eqref{eq:simple-dispersion}.  The chosen neighbourhood of \(c_0\)
does not contain sonic branch points or critical singularities on the real
\(y\)-axis; if a critical point is present, the standard causal indentation
is used.  The admittance perturbation is measured by \(A(\alpha,c_0)\).  If
the physical wall law has the form \(A=\varepsilon\mathcal A(\alpha,c)\),
with \(\varepsilon\ll1\) and \(\mathcal A\) analytic near \(c_0\), then the
\(c\)-dependence of \(\mathcal A\) contributes only to the \(\mathcal O(\varepsilon^2)\)
term in the eigenvalue expansion.
\end{assumption}

\begin{lemma}[The denominator as a dispersion derivative]
\label{lem:I-Dc}
Let \(P(y;c)\) be the shooting solution used in
\S\ref{ssec:dispersion}, normalised by
\[
        P(y_b;c)=1,\qquad P'(y_b;c)=-\gamma(c),
\]
and let $D_0(c):=D(c,0)=P'(0;c)$
be the rigid-wall dispersion function.  Suppose that
\(D_0(c_0)=0\), and set
$\hat p_0(y)=P(y;c_0)$.
Then $\displaystyle D_0'(c_0)
        =
        -\,\frac{c_0^2}{T_w\hat p_0(0)}\,I$,
where \(I\) is the denominator defined in \eqref{eq:Idef}.  In particular,
        $I=0$ if and only if
        $D_0'(c_0)=0$.
Thus, if \(c_0\) is a simple zero of \(D_0\), then \(I\neq0\).
\end{lemma}

\begin{proof}
Differentiate the self-adjoint Rayleigh equation
\[
        (\Phi \hat p')'=\alpha^2(\Phi-M^2)\hat p,
        \qquad
        \Phi(y;c)=\frac{T}{(U-c)^2},
\]
with respect to \(c\) at \(c=c_0\).  Write
        $q(y)=\pd{P}{c}(y;c_0)$ and 
        $\Phi_0(y)=\Phi(y;c_0)$.
Since \(\hat p_0=P(\cdot;c_0)\), the differentiated equation is
\begin{equation}
        (\Phi_0 q')'
        -
        \alpha^2(\Phi_0-M^2)q
        =
        -\left[
        \left(\pd{\Phi}{c}\Big|_{c_0}\hat p_0'\right)'
        -
        \alpha^2\pd{\Phi}{c}\Big|_{c_0}\hat p_0
        \right].
\label{eq:q-eqn}
\end{equation}
Multiply \eqref{eq:q-eqn} by \(\hat p_0\), integrate over
\((0,y_b)\), and integrate by parts on both sides, using the homogeneous equation
satisfied by \(\hat p_0\).  This yields the exact Green's identity for the parameter derivative:
\begin{equation}
        \Big[
        \Phi_0
        \big(
        \hat p_0 q'-\hat p_0' q
        \big)
        +
        \pd{\Phi}{c}\Big|_{c_0}\hat p_0\,\hat p_0'
        \Big]_{0}^{y_b}
        =
        \int_0^{y_b}
        \pd{\Phi}{c}\Big|_{c_0}
        \left(
        \hat p_0'^{\,2}
        +
        \alpha^2\hat p_0^{\,2}
        \right)\dd y .
\label{eq:lemma-full}
\end{equation}

We now evaluate the endpoint terms on the left-hand side.  At the rigid wall,
\[
        \hat p_0'(0)=0,\qquad
        q'(0)=\pd{}{c}P'(0;c)\Big|_{c=c_0}=D_0'(c_0).
\]
Because \(\hat p_0'(0)=0\), the term proportional to \(\partial\Phi/\partial c\) vanishes identically, hence the lower endpoint contributes
$\displaystyle -\Phi_0(0)\hat p_0(0)q'(0)
        =
        -\frac{T_w}{c_0^2}\hat p_0(0)D_0'(c_0)$.

At \(y_b\), the shooting normalisation gives \(P(y_b;c)=1\), so
\(\hat p_0(y_b)=1\) and \(q(y_b)=0\).  Moreover \(P'(y_b;c)=-\gamma(c)\), so
\(\hat p_0'(y_b)=-\gamma_0\) and \(q'(y_b)=-\gamma'(c_0)\).
Direct substitution into the upper boundary term yields exactly \(-I_\infty\):
\[
        -\Phi_0(y_b)\gamma'(c_0)
        -
        \pd{\Phi}{c}\Big|_{c_0}(y_b)\gamma_0
        =
        -\frac{1}{(1-c_0)^2}\frac{\alpha^2-\gamma_0^2}{\gamma_0(1-c_0)}
        -
        \frac{2}{(1-c_0)^3}\gamma_0
        =
        -\frac{\gamma_0^2+\alpha^2}{\gamma_0(1-c_0)^3}
        =
        -I_\infty .
\]
Moving this upper endpoint contribution \(-I_\infty\) to the right-hand side of \eqref{eq:lemma-full} completes the exact definition of \(I\).
Therefore \eqref{eq:lemma-full} reduces to
$\displaystyle -\frac{T_w}{c_0^2}\hat p_0(0)D_0'(c_0)
        =
        I$.
Equivalently,
$\displaystyle D_0'(c_0)
        =
        -\,\frac{c_0^2}{T_w\hat p_0(0)}\,I$.
If instead the pressure eigenfunction is normalised by the wall condition
\(\hat p_0(0)=1\), the same calculation gives the same equivalence with the
corresponding non-zero normalisation factor.  In either normalisation,
        \(I=0\) if and only if
        \(D_0'(c_0)=0\).
Thus a simple zero of the rigid dispersion relation has \(I\neq0\).
\end{proof}

\begin{theorem}[Admittance sensitivity of a compressible Rayleigh mode]
\label{thm:sens}
Let $(\hat p_0,c_0)$ be a rigid-wall eigenpair of
\eqref{eq:rayleigh}, \eqref{eq:farfield} satisfying $\hat p_0'(0)=0$, under
Assumption~\ref{ass:spectral}, and let the wall condition be perturbed to
\eqref{eq:wallBC}.  Then, for sufficiently small $|A|$,
\begin{equation}
c(A)=c_0+K A+\mathcal O(|A|^2),
\qquad
K=\frac{\ii\alpha\,\hat p_0(0)^2}{c_0\,I},
\label{eq:thm}
\end{equation}
where
\begin{equation}
I=\int_0^\infty\frac{2T}{(U-c_0)^3}
\Big(\hat p_0'^{\,2}+\alpha^2\hat p_0^{\,2}\Big)\,\dd y
\;+\;I_\infty,
\qquad
I_\infty=\frac{2}{(1-c_0)^3}\,
\frac{(\gamma_0^2+\alpha^2)\,\hat p_0(y_b)^2}{2\gamma_0},
\label{eq:Idef}
\end{equation}
the integral being taken over the boundary layer up to $y_b$ and the tail
term $I_\infty$ accounting exactly for the uniform-stream continuation
\eqref{eq:farfield}.  Consequently
\begin{equation}
\delta\sigma=\alpha\,\Imag(KA)+\mathcal O(|A|^2).
\label{eq:dsig}
\end{equation}
\end{theorem}

\begin{proof}
By \S\ref{ssec:dispersion}, the rigid-wall eigenvalue is a simple zero of
the analytic dispersion function,
\[
        D(c_0,0)=0,\qquad D_c(c_0,0)\ne0 .
\]
Hence the implicit function theorem gives a unique analytic eigenvalue
branch \(c=c(A)\), for \(|A|\) sufficiently small, satisfying
\(D(c(A),A)=0\) and \(c(0)=c_0\).
It remains to compute the first derivative of this branch at \(A=0\).
We do this by Green's identity, in a form independent of the normalisation
of the eigenfunction.

Let \(\hat p=\hat p_0+\delta\hat p\) and
        \(c=c_0+\delta c\),
where \((\hat p,c)\) lies on the analytic branch.  Write
\[
        \Phi_0(y)=\Phi(y;c_0),\qquad
        \Phi(y)=\Phi(y;c)=\frac{T(y)}{(U(y)-c)^2}.
\]
The rigid and perturbed eigenfunctions satisfy
\begin{equation}
(\Phi_0\hat p_0')'=\alpha^2(\Phi_0-M^2)\hat p_0,
\qquad
(\Phi\hat p')'=\alpha^2(\Phi-M^2)\hat p .
\label{eq:pair}
\end{equation}
Multiply the first equation in \eqref{eq:pair} by \(\hat p\), the second by
\(\hat p_0\), subtract, and integrate from \(0\) to \(y_b\).  Expanding the integrand, keeping terms up to first order in the perturbations \(\delta c\) and \(\delta\hat p\), and carefully retaining all first-order boundary variations gives
\begin{equation}
\Big[
\Phi_0
\big(
\hat p_0\,\delta\hat p'
-
\hat p_0'\,\delta\hat p
\big)
+
\delta\Phi\,\hat p_0\,\hat p_0'
\Big]_{0}^{y_b}
=
\delta c
\int_0^{y_b}
\pd{\Phi}{c}\Big|_{c_0}
\left(
\hat p_0'^{\,2}
+
\alpha^2\hat p_0^{\,2}
\right)\dd y
+
\mathcal O(|\delta c|^2+|\delta c|\,\|\delta\hat p\|).
\label{eq:identity}
\end{equation}
Here \(\displaystyle  \delta\Phi = \delta c \pd{\Phi}{c}\big|_{c_0} + \mathcal O(|\delta c|^2)\) with
$\displaystyle \pd{\Phi}{c}\Big|_{c_0}=
\frac{2T}{(U-c_0)^3}$,
and products such as \(\hat p\hat p_0\) and \(\hat p'\hat p_0'\) have been
replaced by \(\hat p_0^2\) and \(\hat p_0'^{\,2}\) to first order.

We now evaluate the boundary terms.  Set
\[
        B(y)
        =
        \Phi_0
        \big(
        \hat p_0\,\delta\hat p'
        -
        \hat p_0'\,\delta\hat p
        \big)
        +
        \delta\Phi\,\hat p_0\,\hat p_0',
\]
so that the left-hand side of \eqref{eq:identity} is \(B(y_b)-B(0)\).
At the wall, \(\hat p_0'(0)=0\).  The perturbed wall condition
\eqref{eq:wallBC} gives
\[
        \hat p'(0)
        =
        -\frac{\ii\alpha cA}{T_w}\hat p(0).
\]
Since \(\hat p'(0) = \hat p_0'(0) + \delta\hat p'(0)\), \(\hat p(0) = \hat p_0(0) + \delta\hat p(0)\), and \(c = c_0 + \delta c\), expanding the boundary condition yields
\[
        \delta\hat p'(0)
        =
        -\frac{\ii\alpha c_0A}{T_w}\hat p_0(0)
        +
        \mathcal O(|A|\,|\delta c| + |A|\,|\delta\hat p(0)|).
\]
Consequently, noting that the \(\delta\Phi\) term in \(B(0)\) vanishes identically because \(\hat p_0'(0)=0\), we have
\[
        -B(0)
        =
        -\Phi_0(0)\big(\hat p_0(0)\delta\hat p'(0) - \hat p_0'(0)\delta\hat p(0)\big)
        =
        \frac{\ii\alpha A}{c_0}\hat p_0(0)^2
        +
        \mathcal O(|A|\,|\delta c| + |A|\,|\delta\hat p(0)|),
\]
because \(U(0)=0\) and hence \(\Phi_0(0)=T_w/c_0^2\).

It remains to account for the upper endpoint.  At \(y=y_b\) the solution is
matched to the uniform-stream tail
\[
        \hat p_0(y)=\hat p_0(y_b)\exp[-\gamma_0(y-y_b)],
        \qquad
        \gamma_0=\alpha[1-M^2(1-c_0)^2]^{1/2},
        \qquad \Real\gamma_0>0 .
\]
The variation of the profile \(\Phi(y;c)\) and the decay exponent \(\gamma(c)\) contributes to the boundary term at \(y_b\). Explicit evaluation of \(B(y_b)\) by matching with the perturbed tail exactly yields \(B(y_b)=-\delta c\,I_\infty+\mathcal O(|\delta c|^2)\). Equivalently, and more invariantly, this total boundary contribution is identical to the closed-form continuation of the integral in \eqref{eq:identity} over the uniform stream \(y>y_b\).  Since \(U=T=1\) there,
\[
        \pd{\Phi}{c}\Big|_{c_0}
        =
        \frac{2}{(1-c_0)^3},
        \qquad
        \hat p_0'^{\,2}+\alpha^2\hat p_0^{\,2}
        =
        (\gamma_0^2+\alpha^2)\hat p_0(y_b)^2
        \exp[-2\gamma_0(y-y_b)] .
\]
Thus
\[
\int_{y_b}^{\infty}
\pd{\Phi}{c}\Big|_{c_0}
\left(
\hat p_0'^{\,2}
+
\alpha^2\hat p_0^{\,2}
\right)\dd y
=
\frac{2}{(1-c_0)^3}
\frac{(\gamma_0^2+\alpha^2)\hat p_0(y_b)^2}{2\gamma_0}
=
I_\infty .
\]
Substituting the wall contribution \(B(0)\) and the tail contribution \(B(y_b)\) into
\eqref{eq:identity} and moving the \(I_\infty\) term to the right-hand side gives
\[
        \frac{\ii\alpha A}{c_0}\hat p_0(0)^2
        =
        \delta c
        \left[
        \int_0^{y_b}
        \frac{2T}{(U-c_0)^3}
        \left(
        \hat p_0'^{\,2}
        +
        \alpha^2\hat p_0^{\,2}
        \right)\dd y
        +
        I_\infty
        \right]
        +
        \mathcal O(|\delta c|^2+|\delta c|\,\|\delta\hat p\|+|A|\,|\delta c|+|A|\,|\delta\hat p(0)|).
\]
By the analytic branch construction, both \(\delta c=\mathcal O(A)\) and \(\delta\hat p=\mathcal O(A)\).  Therefore the entire remainder reduces strictly to \(\mathcal O(|A|^2)\), and
\[
        \delta c
        =
        \frac{\ii\alpha A\hat p_0(0)^2}{c_0 I}
        +
        \mathcal O(|A|^2),
\]
where \(I\) is defined by \eqref{eq:Idef}.  Hence
\[
        c(A)=c_0+
        \frac{\ii\alpha\hat p_0(0)^2}{c_0 I}\,A
        +
        \mathcal O(|A|^2).
\]

Finally, \(I\ne0\).  Indeed, in the shooting formulation of
\S\ref{ssec:dispersion}, the same Green identity identifies the denominator
\(I\), up to the non-zero normalisation factor fixed by the shooting
solution, with the derivative \(D_c(c_0,0)\) of the rigid-wall dispersion
function.  Since the eigenvalue is simple, \(D_c(c_0,0)\ne0\), and therefore
\(I\ne0\).  Taking imaginary parts and multiplying by \(\alpha\) gives
\(\delta\sigma
        =
        \alpha\,\Imag(KA)
        +
        \mathcal O(|A|^2)\),
which proves \eqref{eq:dsig}.
\end{proof}

\begin{corollary}[Phase criterion]
\label{cor:phase}
Write $A=|A|\ee^{\ii\varphi}$ and $K=|K|\ee^{\ii\varphi_K}$.  Then, to first
order, $\delta\sigma=\alpha|K||A|\sin(\varphi+\varphi_K)$: the admittance
stabilises the mode if and only if $\Imag(K\ee^{\ii\varphi})<0$, i.e.\ if
and only if its phase lies in the open half-plane
$\varphi\in(-\pi-\varphi_K,\,-\varphi_K)$ (mod $2\pi$).
\end{corollary}

As demonstrated in Figure~\ref{fig:spectral_sensitivity}, the spectral sensitivity theorem elegantly decouples the outer mode physics from the inner wall mechanisms (Section~\ref{sec:mechanisms}) via a simple inner product.
Proposition~\ref{prop:phase-geometry} rewrites this condition in Cartesian
form and gives the special cases used below for resistive, positive-reactive
and negative-reactive wall responses.

\begin{figure}[htbp]
    \centering
    \includegraphics[width=0.9\textwidth]{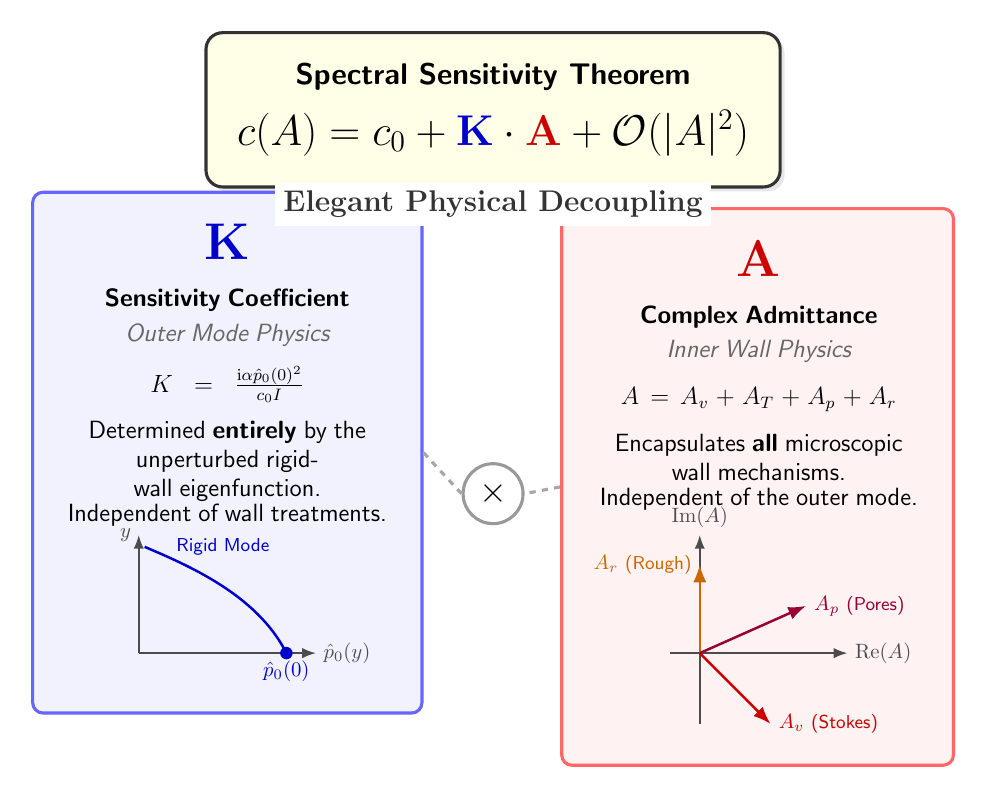} 
    \caption{Visual breakdown of the spectral sensitivity theorem, illustrating the decoupling between the sensitivity coefficient (outer mode physics) and the complex admittance (inner wall physics).}
    \label{fig:spectral_sensitivity}
\end{figure}

\begin{proposition}[Geometry of the stabilising phase band]
\label{prop:phase-geometry}
Let \(K=K_R+\ii K_I\neq0\), and let
$A=A_R+\ii A_I=|A|\ee^{\ii\varphi}$.
To first order in \(|A|\),
$        \delta\sigma
        =
        \alpha\Imag(KA)
        =
        \alpha(K_R A_I+K_I A_R)$.
Hence the stabilising admittances form the open half-plane
\begin{equation}
        K_R A_I+K_I A_R<0 .
\label{eq:stab-half-plane}
\end{equation}
In particular:
\[
        \text{purely resistive walls }(A>0)
        \text{ stabilise if and only if } K_I<0,
\]
\[
        \text{positive reactance }(A=\ii a,\ a>0)
        \text{ stabilises if and only if } K_R<0,
\]
and
\[
        \text{negative reactance }(A=-\ii a,\ a>0)
        \text{ stabilises if and only if } K_R>0.
\]
For passive walls, \(\Real A\ge0\), or equivalently
\(\varphi\in[-\pi/2,\pi/2]\), the passive sector has angular width \(\pi\).
Therefore, except in the degenerate case in which one boundary ray is
neutral, no open stabilising half-plane can contain the entire passive
sector.  Passive walls are consequently not generically stabilising: a
passive wall may damp or amplify a given mode depending on the phase of its
admittance relative to \(K\).
\end{proposition}

\begin{proof}
The first identity follows directly from Corollary~\ref{cor:phase}:
\[
        \Imag(KA)
        =
        \Imag\!\left[(K_R+\ii K_I)(A_R+\ii A_I)\right]
        =
        K_R A_I+K_I A_R .
\]
The stabilising set is therefore the open half-plane
\eqref{eq:stab-half-plane}.  If \(A>0\) is purely resistive, then
$\Imag(KA)=A K_I$,
so stabilisation is equivalent to \(K_I<0\).  If \(A=\ii a\), \(a>0\), then
$\Imag(KA)=aK_R$,
so positive reactance stabilises if and only if \(K_R<0\).  If
\(A=-\ii a\), \(a>0\), then
$\Imag(KA)=-aK_R$,
so negative reactance stabilises if and only if \(K_R>0\).

Finally, the passive sector \(\Real A\ge0\) is the closed right half-plane
in the \(A\)-plane and has angular width \(\pi\).  The stabilising set is
also a half-plane through the origin.  An open half-plane cannot strictly
contain another closed half-plane of the same angular width; at most their
boundary rays can coincide, in which case one passive boundary direction is
neutral.  Thus passivity alone does not imply stabilisation.
\end{proof} 

As illustrated in Figure~\ref{fig:admittance_plane}, the complex admittance plane is divided into stabilizing and destabilizing regions by a neutral boundary, demonstrating how purely passive or reactive mechanisms can alter mode stability based on their phase.

\begin{figure}[htbp]
    \centering
    \includegraphics[width=0.85\textwidth]{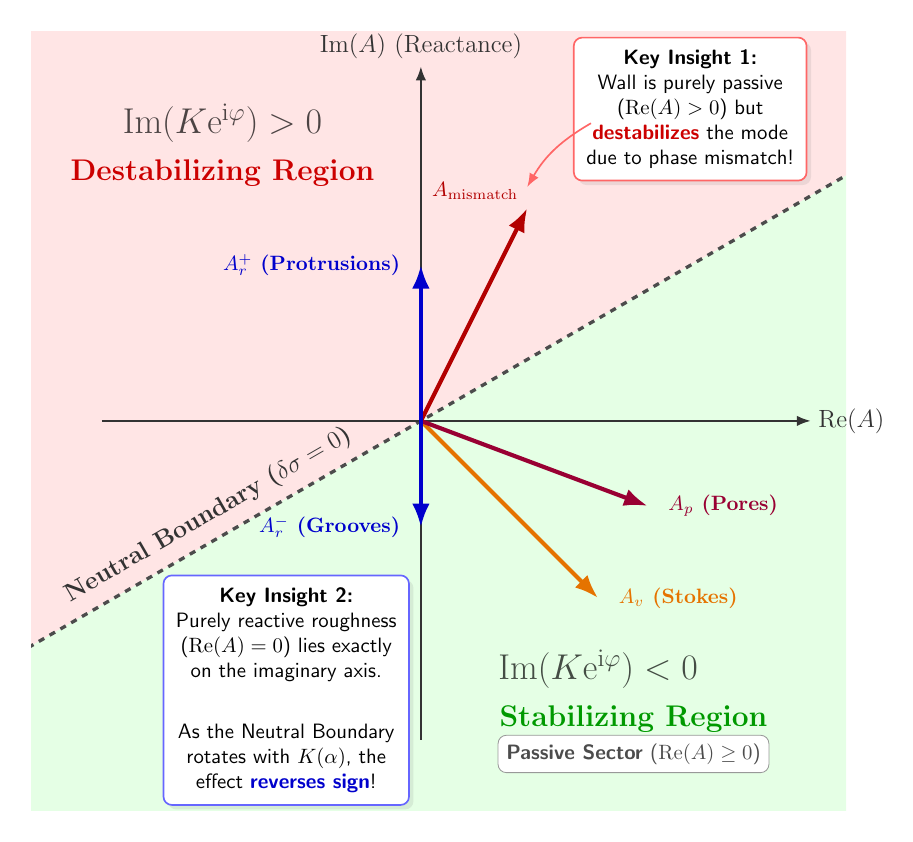} 
    \caption{Schematic of the complex admittance plane highlighting the stabilizing and destabilizing regions separated by the neutral boundary, along with the phase effects of various wall mechanisms.}
    \label{fig:admittance_plane}
\end{figure}

\begin{remark}[Locality and near-degeneracy]
\label{rem:local}
The expansion \eqref{eq:thm} is local: it tracks the simple eigenvalue
branch through \(c_0\) and says nothing about other modes elsewhere in the
spectrum.  Both the constant in the \(\mathcal O(|A|^2)\) remainder and the radius of
validity degenerate as the eigenvalue approaches another branch.  Near the
synchronisation region of the Mach-4.5 flow
(\(\alpha\approx0.26\), \S\ref{ssec:band}), where two branches nearly
coalesce, \(|K|\) is observed to spike and the scalar sensitivity formula
should be interpreted only as an indicator of the loss of simple-eigenvalue
regularity.  Numerical values within a small neighbourhood of this region
are shown only to indicate this breakdown; they are not used as validation
of Theorem~\ref{thm:sens}.
\end{remark}

\begin{remark}[Use at $\mathcal O(1)$ admittance]
\label{rem:O1}
Real coatings, and the viscous wall layer at moderate Reynolds number, have
$|A|=\mathcal O(1)$.  There the formula \eqref{eq:thm} is \emph{not} invoked as an
asymptotic approximation: it serves as a phase diagnostic and as an
a-posteriori comparison, while all quantitative results are obtained from
the full eigenvalue problem retaining \eqref{eq:wallBC} exactly
(\S\ref{sec:numerics}).
\end{remark}

\section{Effective admittances of thin wall mechanisms}
\label{sec:mechanisms}

\begin{assumption}[Thin wall-layer reduction]
\label{ass:thin}
Each wall mechanism occupies a region of thickness $\delta_m$ satisfying
\begin{equation}
\delta_m\ll \delta_{BL}=\mathcal O(1),
\qquad
\alpha\,\delta_m\ll1 ,
\label{eq:ordering}
\end{equation}
so that the disturbance pressure is uniform across the sublayer to leading
order, $\hat p\simeq\hat p_w$, and the sublayer returns to the outer flow a
wall-normal flux proportional to this common forcing.  Cross-interactions
between distinct mechanisms --- orifice end corrections, modification of the
Stokes layer by coating transpiration, roughness-induced local separation
--- are higher order in the ratios of the layer scales or quadratic in the
individual admittances, so the admittances of coexisting mechanisms are
additive at leading order:
\begin{equation}
A=A_v+A_T+A_p+A_r+\mathcal O(\textnormal{products}).
\label{eq:additive}
\end{equation}
\end{assumption}

\begin{proposition}[Thin-layer admittance reduction]
\label{thm:thin-admittance}
Let a wall mechanism occupy an inner region of dimensional thickness
\(\ell_m\), and write
\[
        \delta_m=\frac{\ell_m}{L^*}\ll1,
        \qquad
        \alpha\delta_m\ll1 .
\]
Introduce the stretched normal coordinate \(Y=y/\delta_m\).  Suppose that,
after nondimensionalisation, the inner disturbance problem has the linear
form
$\mathcal L_m(\alpha,c;\delta_m)\,{\bf U}_m
        =
        \mathcal F_m(\alpha,c;\delta_m)\,\hat p_w$,
where \(\hat p_w=\hat p(0^+)\) is the outer wall pressure, and that the
solution defines an outward inner flux
$q_m=\mathcal Q_m{\bf U}_m$.
Assume that the inner problem is well posed uniformly for \(c\) in the
neighbourhood considered in Assumption~\ref{ass:spectral}, and that its flux
admits the expansion
\begin{equation}
        q_m
        =
        -A_m(\alpha,c)\hat p_w
        +
        \mathcal O\!\left(
        \alpha\delta_m\,|\hat p_w|
        +
        \delta_m\,|\hat p_y(0^+)|
        \right).
\label{eq:inner-flux-law}
\end{equation}
Then the outer problem satisfies the effective admittance condition
\begin{equation}
        \hat v(0^+)
        =
        -A_m(\alpha,c)\hat p(0^+)
        +
        \mathcal O\!\left(
        \alpha\delta_m\,|\hat p(0^+)|
        +
        \delta_m\,|\hat p_y(0^+)|
        \right).
\label{eq:effective-single}
\end{equation}
Equivalently, if the remainder is absorbed into the boundary condition, the
wall admittance is
\[
        A=A_m+\mathcal O(\alpha\delta_m+\delta_m\,|\hat p_y(0^+)|/|\hat p(0^+)|).
\]

If \(N\) independent thin mechanisms are forced by the same leading-order
outer pressure and their fluxes are additive, then
\begin{equation}
        \hat v(0^+)
        =
        -\left(\sum_{m=1}^N A_m\right)\hat p(0^+)
        +
        \mathcal O\!\left(
        \sum_{m=1}^N
        \alpha\delta_m\,|\hat p(0^+)|
        +
        \sum_{m=1}^N
        \delta_m\,|\hat p_y(0^+)|
        +
        \sum_{m\ne n}|A_mA_n|\,|\hat p(0^+)|
        \right),
\label{eq:effective-many}
\end{equation}
and therefore
\begin{equation}
        A
        =
        \sum_{m=1}^N A_m
        +
        \mathcal O\!\left(
        \sum_{m=1}^N\alpha\delta_m
        +
        \sum_{m=1}^N\delta_m\,|\hat p_y(0^+)|/|\hat p(0^+)|
        +
        \sum_{m\ne n}|A_mA_n|
        \right).
\label{eq:admittance-additive-estimate}
\end{equation}
\end{proposition}

\begin{proof}
The outer pressure has the Taylor expansion, in the inner coordinate
\(Y=y/\delta_m\),
\[
        \hat p(\delta_m Y)
        =
        \hat p(0^+)
        +
        \delta_m Y\,\hat p_y(0^+)
        +
        \mathcal O(\delta_m^2Y^2).
\]
Moreover, the streamwise phase factor varies across the layer by
\[
        \ee^{\ii\alpha x}
        \quad\hbox{with normal correction of size}\quad
        \mathcal O(\alpha\delta_m),
\]
so the leading forcing of the inner problem is the constant pressure
\(\hat p_w=\hat p(0^+)\).  By the assumed uniform well-posedness and the
linear flux law \eqref{eq:inner-flux-law}, the inner solution returns the
outward flux
\[
        q_m
        =
        -A_m(\alpha,c)\hat p(0^+)
        +
        \mathcal O\!\left(
        \alpha\delta_m\,|\hat p(0^+)|
        +
        \delta_m\,|\hat p_y(0^+)|
        \right).
\]
Matching of normal velocity between the inner and outer regions gives
$\hat v(0^+)=q_m$,
which proves \eqref{eq:effective-single}.  For several mechanisms occupying
asymptotically separated or independently homogenised thin regions, the
leading pressure forcing is common to all of them, while conservation of
mass makes the returned normal flux the sum of the individual fluxes:
\[
        \hat v(0^+)=\sum_{m=1}^N q_m .
\]
Adding the single-mechanism expansions gives the first two error terms in
\eqref{eq:effective-many}.  The product terms \(A_mA_n\) represent the first
neglected interactions: for example, modification of one sublayer by the
transpiration induced by another, or the displacement of one effective
boundary by another.  These interactions are beyond the leading matching
order retained here.  This gives \eqref{eq:effective-many} and hence
\eqref{eq:admittance-additive-estimate}.
\end{proof}

Proposition~\ref{thm:thin-admittance} formalises the pressure-driven,
velocity-returned matching structure used below.  The rest of this section
computes the leading coefficients \(A_m\) for the particular mechanisms of
interest. 
As illustrated in Figure~\ref{fig:linear_additivity}, the linear additivity of different wall mechanisms allows engineers to selectively combine treatments, effectively steering the total admittance vector into the stabilizing region.
Figure~\ref{fig:unified_encapsulation} further shows that various microscopic mechanisms at the wall, including roughness, viscous/thermal sublayers, and porous coatings, can be encapsulated into a unified macroscopic admittance interface.

\begin{figure}[htbp]
    \centering
    \includegraphics[width=0.85\textwidth]{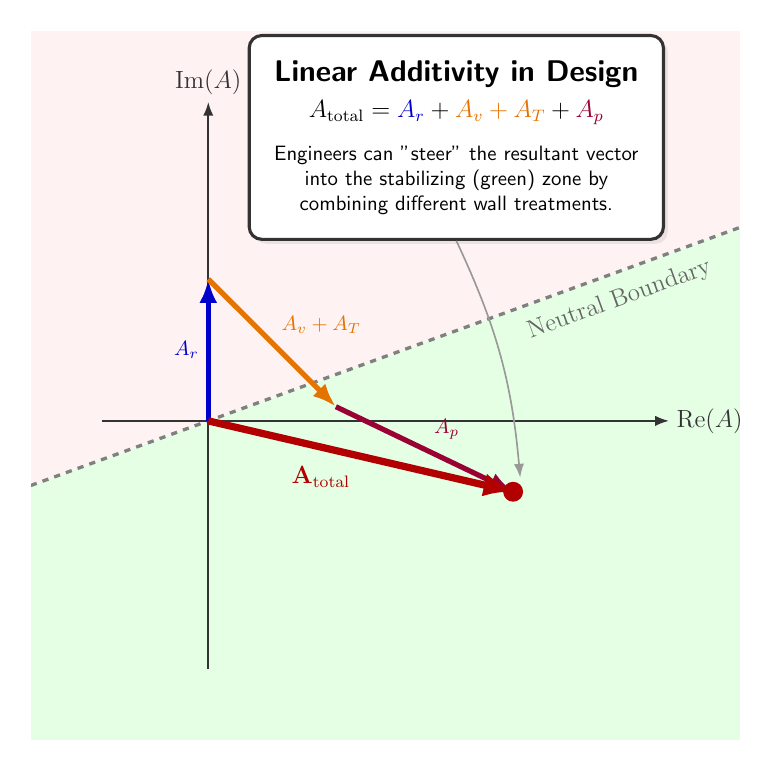} 
    \caption{Vector addition of individual admittance components in the complex plane, demonstrating how combined wall treatments can achieve a stabilizing total effective admittance.}
    \label{fig:linear_additivity}
\end{figure}

\begin{figure}[htbp]
    \centering
    \includegraphics[width=0.85\textwidth]{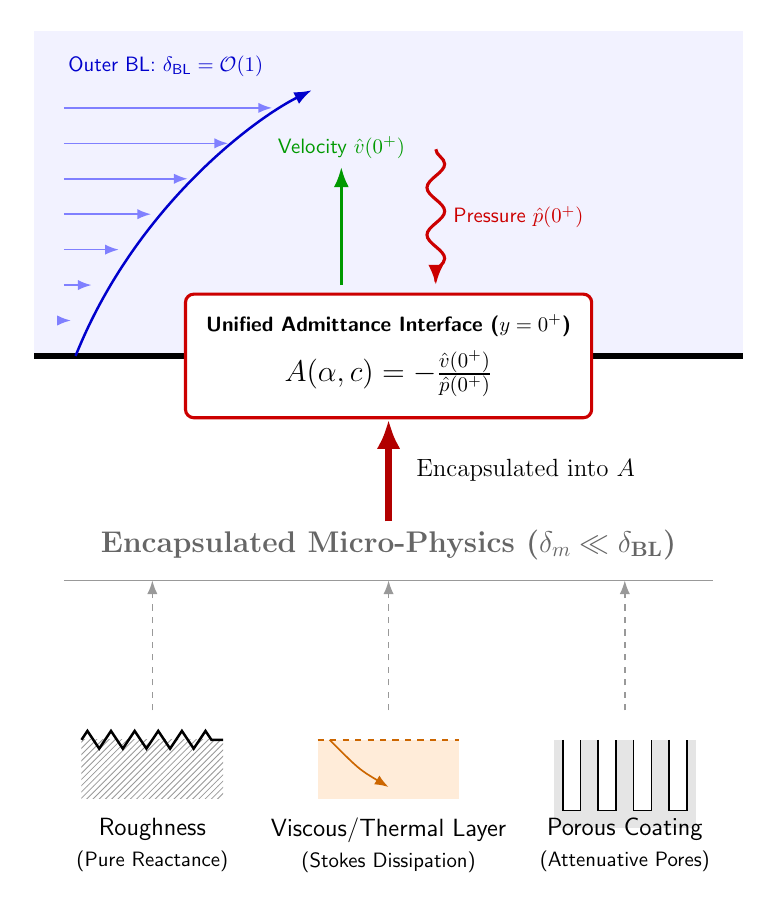} 
    \caption{Conceptual diagram illustrating the encapsulation of underlying micro-physics into a unified effective admittance interface for the outer boundary layer.}
    \label{fig:unified_encapsulation}
\end{figure}

\subsection{The viscous Stokes and thermal sublayers}
\label{ssec:stokes}

For an inviscidly controlled mode of frequency $\omega=\alpha c$, viscosity
acts at the wall in a Stokes layer of thickness
$\delta_v=(\nu_w/\omega)^{1/2}=T_w(\omega\Rey)^{-1/2}$, with
$\nu_w=T_w^2/\Rey$ under the Chapman law; the ordering
\eqref{eq:ordering} holds with $\delta_m=\delta_v$.  Provided
$\omega\gg\alpha\lambda_w\delta_v$ --- comfortably satisfied for the second
mode; the uniformly valid Airy-function (Tietjens) form, and its
triple-deck connection in the opposite limit, are given in
Appendix~\ref{app:airy} --- the base shear may be neglected across the
layer. The Stokes form is the high-frequency limit of the shear-affected wall
layer.  To see the ordering explicitly, compare the change of base velocity
across the Stokes layer with the phase speed:
\[
        \epsilon_A
        =
        \frac{\alpha\lambda_w\delta_v}{\omega}
        =
        \frac{\lambda_w\delta_v}{c},
        \qquad
        \delta_v=\frac{T_w}{(\alpha c\Rey)^{1/2}},
        \qquad
        \omega=\alpha c .
\]
Thus $\displaystyle \epsilon_A
        =
        \frac{\lambda_w T_w}{c(\alpha c\Rey)^{1/2}}$.
When \(|\epsilon_A|\ll1\), the base shear is a perturbation across the
Stokes layer and the uniformly valid Airy-function admittance of
Appendix~\ref{app:airy} reduces to
$A_v^{\rm Airy}
        =
        A_v\left(1+\mathcal O(\epsilon_A)\right)$.
At the reference second-mode point
\[
        \alpha=0.31,\qquad
        c_0=0.9082+0.0295\ii,\qquad
        \Rey=2000,\qquad
        \lambda_w=0.0930,\qquad
        T_w=5.05,
\]
we obtain
\[
        |\epsilon_A|
        =
        \left|
        \frac{\lambda_w T_w}
        {c_0(\alpha c_0\Rey)^{1/2}}
        \right|
        \simeq 2.18\times10^{-2}.
\]
The use of the Stokes admittance \eqref{eq:Av} for the second mode therefore
has an expected shear-correction error of only a few per cent.  For lower
frequencies or near the first-mode lower branch, where \(|\epsilon_A|\) is
not small, the Airy admittance \eqref{eq:AvAiry} should be used instead.

The outer solution possesses the wall slip velocity
$\hat u_s=T_w\hat p_w/c$ (from the inviscid $x$-momentum equation at
$y\to0$); the Stokes-layer solution satisfying no slip and matching it is
$\hat u=\hat u_s(1-\ee^{-my})$ with
$m=(\omega/\nu_w)^{1/2}\ee^{-\ii\pi/4}$, and continuity converts the
mass-flux deficit into an outward velocity
$\hat v(0^+)=\ii\alpha\hat u_s/m$, whence
\begin{equation}
A_v=\ee^{-\ii\pi/4}\,T_w^{2}\left(\frac{\alpha}{\Rey c^{3}}\right)^{1/2}.
\label{eq:Av}
\end{equation}
Independently, the outer flow carries an isentropic temperature fluctuation
$\hat\theta_s=(\gamma-1)M^2T_w\hat p_w$ at the wall, whereas the solid wall
is isothermal to fluctuations; the resulting thermal layer of decay rate
$m\sqrt{Pr}$ stores and releases mass through the equation of state, giving
\begin{equation}
A_T=\ee^{-\ii\pi/4}\,\frac{(\gamma-1)M^2T_w}{\sqrt{Pr}}
\left(\frac{\alpha c}{\Rey}\right)^{1/2}.
\label{eq:AT}
\end{equation}
Both carry the Stokes phase $-45^\circ$; their passivity follows from the
inner dissipation identities in Proposition~\ref{prop:passivity}.  Although
asymptotically $\mathcal O(\Rey^{-1/2})$, their numerical prefactor is large at high
Mach number because of the factors $T_w^2$ and $(\gamma-1)M^2T_w$: at
$\Rey=2000$, $\alpha=0.31$, $c=c_0$ we find $A_v+A_T=0.66-0.66\ii$ ---
\emph{not} a small perturbation, and treated accordingly
(Remark~\ref{rem:O1}).  
Viscosity, in this description, is simply another
wall impedance, whose phase locates it in the admittance plane alongside
coatings and roughness (figure~\ref{fig:1}b).

\subsection{Porous coatings of blind pores}
\label{ssec:porous}

Consider a coating of blind cylindrical pores of radius $r_p$, depth $h$ and
surface porosity $\phi$, with $r_p, \delta_v\lesssim1\ll2\pi/\alpha$, so
that \eqref{eq:ordering} holds at the pore scale and the flow in each pore
is one-dimensional and driven by $\hat p_w$.  Solving the oscillatory
pipe-flow (Womersley) and conduction problems across the pore
(Appendix~\ref{app:pore}) yields the classical Zwikker--Kosten effective
medium \citep{zwikker1949sound}: with
$\zeta=r_p(\omega/\nu_w)^{1/2}\ee^{\ii\pi/4}$ and $H(z)=2J_1(z)/[zJ_0(z)]$,
\begin{equation}
\rho_d=\frac{\bar\rho_w}{1-H(\zeta)},
\qquad
C_d=M^2\Big[1+(\gamma-1)H\big(\zeta\sqrt{Pr}\big)\Big],
\label{eq:ZK}
\end{equation}
the dynamic density and (volumetric) compressibility,
$\hat\rho/\bar\rho=C_d\hat p$, interpolating between adiabatic
($H\to0$) and isothermal ($H\to1$) sound in the pore gas.  With pore
wavenumber and characteristic admittance
$\Lambda=\omega(\rho_dC_d)^{1/2}$, $Y_0=(C_d/\rho_d)^{1/2}$, the input
admittance of a blind pore of depth $h$, times the porosity, is
\begin{equation}
A_p=-\ii\,\phi\,Y_0\tan(\Lambda h).
\label{eq:Ap}
\end{equation}

Viscothermal damping selects the decaying pore branch; the corresponding positive-real property of \(A_p\) is verified energetically in Proposition~\ref{prop:passivity}.  The shallow limit is the cavity
compliance $A_p\simeq-\ii\omega\phi C_dh$; strongly damped (narrow) pores
saturate at the semi-infinite-absorber value $A_p\to\phi Y_0$ for
$h\gtrsim1.5$; between these limits lie quarter-wave resonances
$\Real\Lambda\,h\approx(n+\tfrac12)\pi$ at which $|A_p|$ peaks ---
prominent for wide pores, washed out for narrow ones
(\S\ref{ssec:mech-results}).  At the reference conditions
($\phi=0.3$, $r_p=0.2$, $h=2$, $\Rey=2000$) the coating has
$A_p=0.905-0.740\ii$, again $\mathcal O(1)$.

\subsection{Shallow, dense, non-separating roughness}
\label{ssec:rough}

We next consider distributed roughness whose height and spacing are small
compared with the outer boundary-layer scale and whose local flow remains
attached.  Let \(b\) denote the typical roughness amplitude, \(k_r\) a
representative roughness wavenumber, and \(y_e\) the effective protrusion
height measured from the chosen reference plane.  The regime considered here
is
\begin{equation}
        k_r b\ll1,\qquad
        b\ll\delta_v,\qquad
        \alpha b\ll1,\qquad
        Re_b\equiv\frac{\rho_w U_{\rm rel} b}{\mu_w}\ll1 ,
\label{eq:rough-validity}
\end{equation}
where \(U_{\rm rel}\) is the characteristic oscillatory velocity seen by the
texture.  These assumptions express small slope, confinement inside the
viscous wall layer, small size relative to the instability wavelength, and
absence of local inertial separation.  Thus the result below does not apply
to isolated roughness elements, finite-height trips, separated cavities, or
roughness-induced receptivity mechanisms.

In the regime \eqref{eq:rough-validity}, the classical homogenisation
description of fine riblets and grooves replaces the physical wall by an
equivalent smooth impermeability plane displaced by the protrusion height
\(y_e\) \citep{bechert1989viscous,luchini1991resistance}.  We take \(y_e>0\) when the
effective impermeability plane is displaced toward the fluid, and \(y_e<0\)
for groove-dominated geometries referenced to the land surface.  For a
small-amplitude sinusoidal texture, for example, \(y_e=\mathcal O(k_r b^2)\).  The
outer no-penetration condition is therefore imposed at \(y=y_e\):
$\hat v(y_e)=0$.
Transferring this condition to the reference plane gives
\begin{equation}
        \hat v(0^+)
        =
        -y_e\hat v'(0^+)
        +
        \mathcal O(y_e^2\hat v''(0^+)).
\label{eq:rough-transfer}
\end{equation}

It remains to evaluate \(\hat v'(0^+)\) from the outer inviscid equations.
From the wall-normal momentum equation,
        $\ii\alpha\bar\rho(U-c)\hat v=-\hat p'$ and
        $\bar\rho=T^{-1}$.
At the rigid wall \(U(0)=0\), \(\hat p'(0)=0\), and differentiating the
wall-normal momentum relation gives
\[
        -\ii\alpha\frac{c}{T_w}\hat v'(0)
        =
        -\hat p''(0).
\]
The Rayleigh equation evaluated at the wall, again using
\(\hat p'(0)=0\), gives
\[
        \hat p''(0)
        =
        \alpha^2
        \left(
        1-\frac{M^2c^2}{T_w}
        \right)
        \hat p(0).
\]
Hence
\[
        \hat v'(0)
        =
        -\ii\,\frac{\alpha T_w}{c}
        \left(
        1-\frac{M^2c^2}{T_w}
        \right)\hat p(0).
\]
Substitution into \eqref{eq:rough-transfer}, and comparison with
\(\hat v(0^+)=-A_r\hat p(0^+)\), yields
\[
        A_r
        =
        -\ii\,y_e\,\frac{\alpha T_w}{c}
        \left(
        1-\frac{M^2c^2}{T_w}
        \right)
        +
        \mathcal O(\alpha^2 y_e^2)
        +
        \mathcal O\!\left(\frac{b}{\delta_v}A_v\right).
\]
The first remainder is the next Taylor term in the displacement of the
impermeability condition from \(y=y_e\) to the reference plane.  The second
is the leading dissipative correction associated with the extra Stokes-layer
dissipation over the rough wetted area; it is smaller than the leading
reactive protrusion-height term when \(b/\delta_v\ll1\).  Retaining only the
leading term gives the roughness admittance used below:
\begin{equation}
        A_r
        =
        -\ii\,y_e\,\frac{\alpha T_w}{c}
        \left(
        1-\frac{M^2c^2}{T_w}
        \right).
\label{eq:Ar}
\end{equation}

The bracket in \eqref{eq:Ar} is the wall-normal acoustic
parameter \(1-M_w^2\), where \(M_w=Mc/\sqrt{T_w}\) is the relative Mach
number of the phase speed measured with the wall sound speed.  For the
second Mack mode at \(M=4.5\), \(M_w^2\approx3.3>1\).  Thus protrusion-type
roughness, \(y_e>0\), gives \(A_r=+\ii|a_r|\) to leading order, whereas
groove-type roughness, \(y_e<0\), gives the opposite reactance.  By
Corollary~\ref{cor:phase},
        $\delta\sigma
        =
        \alpha\Imag(KA_r)
        =
        \pm\alpha |a_r|\Real K$,
so the sign of the roughness effect is not universal: it is determined by
the sign of \(\Real K\) for the mode and wavenumber under consideration.
Consequently the roughness effect must reverse along any branch on which
\(\arg K\) crosses \(\mp\pi/2\), a parameter-free prediction tested in
\S\ref{ssec:mech-results}.  The result should therefore be read as a
leading reactive homogenised-wall correction, not as a theory of arbitrary
surface roughness.

\subsection{Passivity of the effective admittances}
\label{ssec:passivity-admittances}

The sign convention in \eqref{eq:Adef-intro} was chosen so that a passive
wall has \(\Real A\ge0\).  We now verify this property directly for the
effective admittances derived above.  The calculation is also a useful check
on the signs in \eqref{eq:Av}--\eqref{eq:Ar}.  For a harmonic disturbance,
the time-averaged power per unit plan area entering the wall is
\[
        \mathcal P_w
        =
        -\frac12\Real\{\hat p_w\hat v_w^*\}
        =
        \frac12|\hat p_w|^2\Real A ,
\]
because \(\hat v_w=-A\hat p_w\).  Thus \(\Real A\ge0\) is equivalent to
non-negative time-averaged energy absorption by the wall.

\begin{proposition}[Passivity of the wall-layer admittances]
\label{prop:passivity}
For the sign convention \(\hat v_w=-A\hat p_w\), the viscous and thermal
wall-layer admittances satisfy
$\Real A_v\ge0$, $\Real A_T\ge0$,
with strict inequality for non-zero wall pressure.  The blind-pore coating
admittance satisfies
$\Real A_p\ge0$
whenever the pore walls are passive and the branch of the pore wavenumber is
chosen with viscothermal decay into the coating.  The shallow roughness
admittance is lossless:
$\Real A_r=0$.
\end{proposition}

\begin{proof}
First consider the viscous Stokes layer.  Let
\[
        \hat u=\hat u_s(1-\ee^{-my}),
        \qquad
        m=\left(\frac{\omega}{\nu_w}\right)^{1/2}\ee^{-\ii\pi/4},
        \qquad
        \hat u_s=\frac{T_w}{c}\hat p_w .
\]
The harmonic kinetic-energy balance for the inner streamwise momentum
equation gives
\[
        \frac12|\hat p_w|^2\Real A_v
        =
        \frac{1}{2\Rey}
        \int_0^\infty
        \mu_w|\hat u_y|^2\,\dd y ,
\]
up to the nondimensional constants already absorbed in the definition of
\(A_v\).  The right-hand side is non-negative and is positive unless the
slip velocity vanishes.  Equivalently, inserting \eqref{eq:Av} directly
gives
\[
        A_v
        =
        \ee^{-\ii\pi/4}
        T_w^2\left(\frac{\alpha}{\Rey c^3}\right)^{1/2},
\]
which has positive real part for the trapped growing modes considered here. 
The thermal layer is identical in
structure: multiplying the inner heat equation by the complex conjugate of
the temperature fluctuation and taking the real part gives
\[
        \frac12|\hat p_w|^2\Real A_T
        =
        \frac{1}{2\Rey Pr}
        \int_0^\infty
        k_w|\hat\theta_y|^2\,\dd y
        \ge0 ,
\]
again up to the nondimensional prefactor used in \eqref{eq:AT}.  Hence
\(\Real A_T\ge0\).

For the pore coating, let \(\Omega_p\) denote one pore and let
\(\langle\hat u\rangle\) be the cross-sectionally averaged axial velocity.
The averaged acoustic energy balance for the Womersley--thermal pore problem
gives
\[
        \frac12|\hat p_w|^2\Real\!\left(\frac{A_p}{\phi}\right)
        =
        \frac12
        \int_{\Omega_p}
        \left[
        \mu_w|\nabla\hat u|^2
        +
        \frac{k_w}{T_w}|\nabla\hat\theta|^2
        \right]\dd V
        \ge0 ,
\]
with the precise nondimensional constants depending only on the pressure,
velocity and temperature scalings.  Multiplication by the positive porosity
\(\phi\) preserves the sign, so \(\Real A_p\ge0\).  This is the energetic
counterpart of the Zwikker--Kosten formula \eqref{eq:Ap}: viscous and
thermal diffusion in the pore make the input admittance positive real for
the decaying branch of the pore transmission line.

Finally, the roughness admittance \eqref{eq:Ar} is
\[
        A_r
        =
        -\ii y_e\frac{\alpha T_w}{c}
        \left(1-\frac{M^2c^2}{T_w}\right).
\]
At leading order the homogenised roughness model only displaces the
effective impermeability plane and contains no dissipative inner mechanism.
For a neutral mode, or for the real-frequency passivity check obtained by
setting \(c_i=0\), this admittance is purely imaginary; hence
        $\Real A_r=0$.
It is therefore a reactive, lossless wall response.  Any positive resistance
associated with real roughness arises at the next order through enhanced
Stokes-layer dissipation over the wetted area, and is not part of the
leading protrusion-height admittance \eqref{eq:Ar}.
\end{proof}

\section{Numerical validation and Mach-4.5 demonstration}
\label{sec:numerics}

\subsection{Method and base flow}
\label{ssec:method}

The eigenvalue problem \eqref{eq:rayleigh}, \eqref{eq:farfield}, \eqref{eq:wallBC} is solved by shooting: a fixed-grid fourth-order
Runge--Kutta march from $y_b=20.8$ ($\approx2\delta_{99}$) to the wall on
$4000$ intervals, with the dispersion function
$F(c)=\hat p'(0)+\ii\alpha cA\hat p(0)/T_w$ driven to zero by damped complex
Newton iteration; for $c$-dependent admittances the eigenvalue and
admittance are iterated to joint self-consistency.  Eigenvalues are
grid- and domain-converged to nine significant figures
(Appendix~\ref{app:num}).  Disturbance-layer Prandtl number is $Pr=0.72$;
the reference Reynolds number is $\Rey=2000$ ($Re_x=2\times10^6$).  The
base flow is shown in figure~\ref{fig:1}(a).

\begin{figure}
\centering
\includegraphics[width=\textwidth]{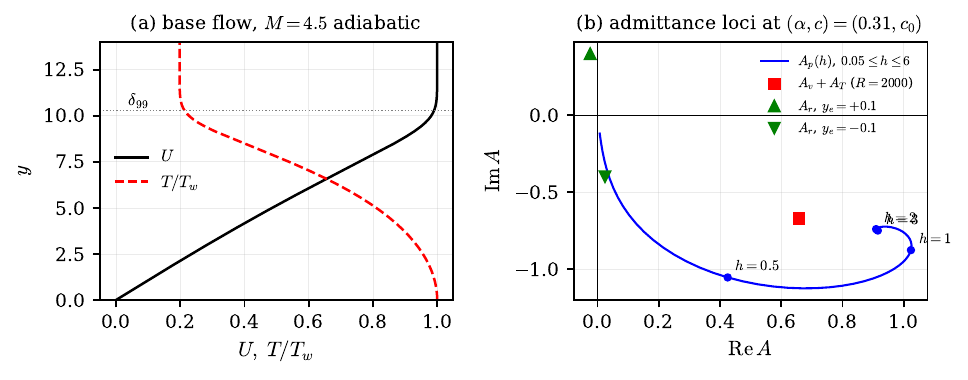}
\caption{(a) Compressible Blasius base flow at $M=4.5$ (adiabatic wall,
Chapman law).  (b) The admittance map at the second-mode reference point
$\alpha=0.31$, $c=c_0=0.9082+0.0295\ii$: locus of the porous admittance
$A_p(h)$, \eqref{eq:Ap}, for $\phi=0.3$, $r_p=0.2$, $\Rey=2000$ as the
depth $h$ varies (the spiral converges to the semi-infinite-absorber point
$\phi Y_0$); viscous admittance $A_v+A_T$ at $\Rey=2000$ (square); reactive
roughness admittances $A_r$ for $y_e=\pm0.1$ (triangles).  Passivity
confines all dissipative elements to $\Real A\ge0$.}
\label{fig:1}
\end{figure}

\subsection{Rigid-wall spectrum}
\label{ssec:rigid}

The rigid adiabatic plate at $M=4.5$ supports two unstable branches
(black curves in figure~\ref{fig:2}).  The low-wavenumber branch
(``mode~I'') is the inflectional inviscid
instability, with peak growth $\sigma=0.0064$ near $\alpha=0.21$.  The
high-wavenumber branch (``mode~II'', the second Mack mode) occupies
$0.255<\alpha<0.47$ and peak growth
$\sigma=0.00918$ at $\alpha=0.315$, $c=0.9067+0.0292\ii$.  Both branches
are trapped, $c_r>1-1/M=0.778$.  Near $\alpha\approx0.26$ the branches
approach (the synchronisation region), where Remark~\ref{rem:local}
applies.

\subsection{Validation of the sensitivity coefficient}
\label{ssec:validation}

Table~\ref{tab:val} compares $K$ from Theorem~\ref{thm:sens} with the
direct quotient $[c(A)-c_0]/A$ computed from the full admittance eigenvalue
problem at $|A|=10^{-3}$, at three second-mode wavenumbers spanning the
band: relative errors are $2.5$--$4.9\times10^{-4}$, consistent with the
$\mathcal O(|A|)$ remainder of the linearisation.  The linear range extends far
beyond infinitesimal admittance: figure~\ref{fig:5}(a) compares the
first-order prediction $\delta\sigma=\alpha\Imag(KA_p)$ with direct
eigenvalues over the full porous-depth sweep, where $|A_p|$ reaches $1.3$;
the formula tracks the direct computation to graphical accuracy, with the
expected mild saturation of the true shift at the largest admittances. The validation points in table~\ref{tab:val} were chosen away from the
closest approach of the two rigid-wall branches. 

\begin{table}[htbp]
\centering
\caption{Validation of Theorem~\ref{thm:sens} at representative second-mode
wavenumbers ($M=4.5$, adiabatic, $\Rey=2000$).  The direct quotient is
computed from the full admittance eigenvalue problem with $A=10^{-3}$.}
\label{tab:val}
\begin{tabular}{ccccc}
\toprule
$\alpha$ & $c_0$ & $K$ from \eqref{eq:thm} & $[c(A)-c_0]/A$ direct &
rel.\ error \\
\midrule
$0.28$ & $0.918421+0.023352\ii$ & $-0.013999-0.007375\ii$ &
$-0.013999-0.007367\ii$ & $4.9\times10^{-4}$ \\
$0.31$ & $0.908217+0.029513\ii$ & $-0.002380-0.007505\ii$ &
$-0.002380-0.007503\ii$ & $2.5\times10^{-4}$ \\
$0.36$ & $0.900805+0.016442\ii$ & $+0.003697-0.000680\ii$ &
$+0.003696-0.000681\ii$ & $3.9\times10^{-4}$ \\
\bottomrule
\end{tabular}
\end{table}

\begin{figure}[htbp]
\centering
\includegraphics[width=\textwidth]{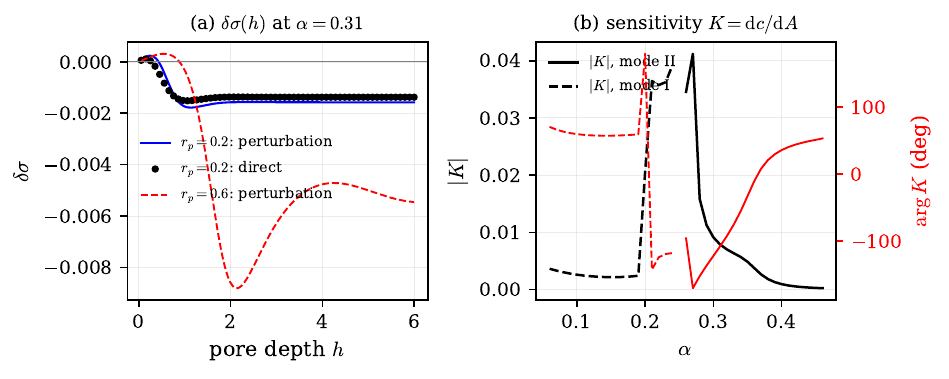}
\caption{(a) Growth-rate shift against pore depth at $\alpha=0.31$:
first-order prediction $\alpha\Imag(KA_p)$ (line) versus direct eigenvalues
(circles) for narrow pores ($r_p=0.2$), and the first-order estimate for
wide pores ($r_p=0.6$, quarter-wave resonances).  (b) Sensitivity
coefficient $K(\alpha)$ along the two rigid branches: modulus (black, left
axis) and phase (red, right axis).}
\label{fig:5}
\end{figure}

\subsection{The admittance map}
\label{ssec:map}

Figure~\ref{fig:1}(b) places the three mechanisms in the complex-$A$ plane
at the second-mode reference point.  The porous locus $A_p(h)$ spirals from
the cavity-compliance regime into the semi-infinite-absorber point
$\phi Y_0$; the viscous point sits at phase $-45^\circ$ with $\mathcal O(1)$
modulus (\S\ref{ssec:stokes}); the roughness reactances lie on the
imaginary axis.  By \eqref{eq:additive}, combined wall treatments are
vector sums on this map, and by Corollary~\ref{cor:phase} their effect on
the mode is read off from the projection onto the direction
$\ii\overline{K}/|K|$.

\subsection{Porous, viscous and rough walls}
\label{ssec:mech-results}

\emph{Porous coatings} (figure~\ref{fig:2}).  In the attenuative regime
considered (narrow pores, $r_p=0.2$), coatings of depth $h=1,2,3$ reduce
the second-mode peak growth from $0.00918$ to $0.00778$--$0.00786$
($14$--$15\%$), the damping concentrated in the band core where the wall
pressure is largest. 
Mode~I responds weakly (peak $0.00637\to0.00616$).
The depth dependence at fixed $\alpha=0.31$
(figure~\ref{fig:5}a) shows the two coating regimes anticipated from
\eqref{eq:Ap}: narrow pores saturate at the depth-insensitive
semi-infinite-absorber level for $h\gtrsim1.5$, while the perturbation
estimate for wide pores ($r_p=0.6$) exhibits quarter-wave resonances at
$h\approx2.0,\,4.2$; near such resonances $|A_p|$ is no longer small and
the first-order formula is a phase diagnostic only (Remark~\ref{rem:O1}).
We emphasise that coating behaviour outside this attenuative regime ---
resonant reinforcement, coating-induced modes --- is documented in the
literature \citep{bres2013second,fedorov2011transition} and is not claimed to be covered
by the present leading-order description.

\begin{figure}[htbp]
\centering
\includegraphics[width=\textwidth]{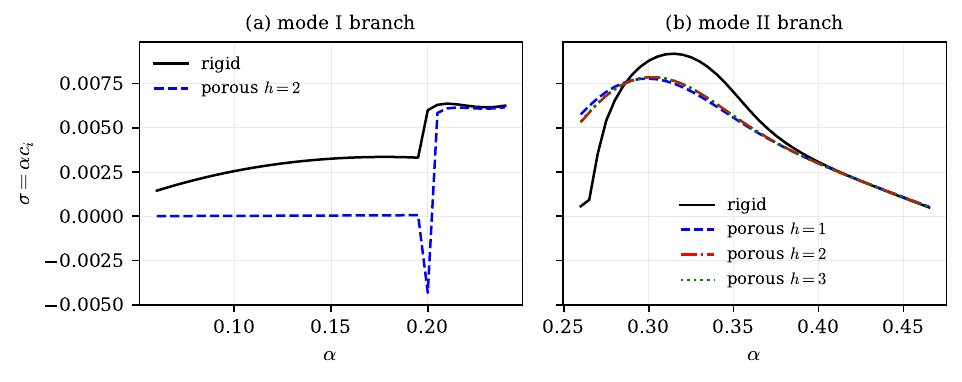}
\caption{Growth rates $\sigma=\alpha c_i$ of (a) the mode-I branch and
(b) the mode-II branch, rigid wall versus porous coatings \eqref{eq:Ap}
with $\phi=0.3$, $r_p=0.2$, $\Rey=2000$, $h=1,2,3$ (attenuative regime).}
\label{fig:2}
\end{figure}

\emph{Viscous wall layers} (figure~\ref{fig:3}).  Closing the inviscid
Rayleigh problem with $A_v+A_T$ at $\Rey=500,2000,8000$ damps the second
mode at all Reynolds numbers --- the Stokes phase $-45^\circ$ lies inside
the stabilising half-plane of Corollary~\ref{cor:phase} with
\eqref{eq:K0} --- the peak growth decreasing linearly in $\Rey^{-1/2}$
(figure~\ref{fig:3}b) as the asymptotics demand, by $17\%$ at $\Rey=500$.
This reproduces, with the mechanism explicit, the classical finding that
the second mode is essentially inviscid with viscosity acting as a
predominantly stabilising correction \citep{mack1984boundary}; for mode~I at low
wavenumber the rotation of $\varphi_K$ (figure~\ref{fig:5}b) moves the
$-45^\circ$ phase toward the destabilising half-plane, the
inviscid-framework counterpart of viscous destabilisation of
Tollmien--Schlichting-type waves \citep{gaponov1975effect,mack1984boundary}.

\begin{figure}
\centering
\includegraphics[width=\textwidth]{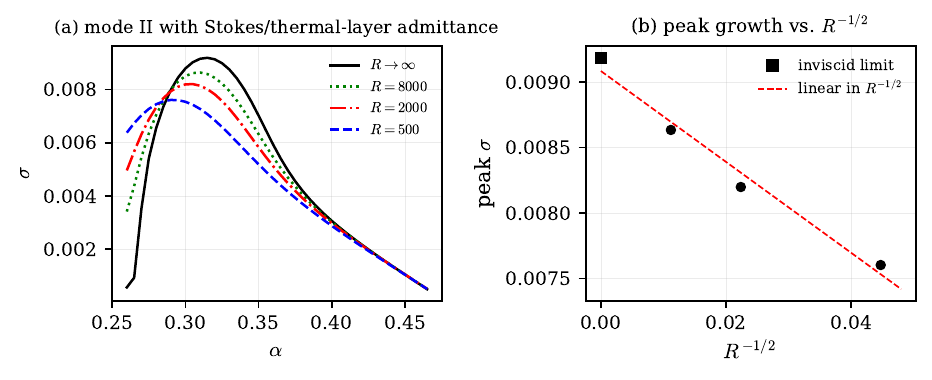}
\caption{(a) Second-mode growth rates with the viscous wall-layer
admittance \eqref{eq:Av}--\eqref{eq:AT} at $\Rey=500,2000,8000$, against
the inviscid (rigid) limit.  (b) Peak growth versus $\Rey^{-1/2}$.}
\label{fig:3}
\end{figure}

\emph{Roughness} (figure~\ref{fig:4}).  The reactive admittance
\eqref{eq:Ar} with $y_e=\pm0.1,\pm0.2$ produces growth-rate shifts linear
in $y_e$ whose \emph{sign reverses along the branch}: protrusions damp the
mode strongly on the lower side of the band (at $\alpha=0.28$,
$c_i=0.0234\to0.0040$ for $y_e=0.2$, the large response reflecting the
spike of $|K|$ near synchronisation) and through the growth peak (peak
$\sigma$ down $3.0\%$ at $y_e=0.2$), but amplify the upper part of the band
($\alpha\gtrsim0.33$); grooves mirror this behaviour.  The crossing
coincides, as Corollary~\ref{cor:phase} requires for a $\pm\ii$ admittance,
with the point where $\arg K$ crosses $-\pi/2$, i.e.\ where $\Real K$
changes sign (figure~\ref{fig:5}b): the computation thus validates not
just the magnitude but the \emph{phase structure} of the sensitivity
coefficient.  Net effect on the band maximum: protrusions reduce, grooves
increase the peak second-mode growth, consistent in sign with wavy-wall
observations \citep{fujii2006experiment,chuvakhov2016spontaneous}; the same reactance with the
mode-I sensitivity changes r\^ole, a reminder that ``roughness'' has no
universal sign \citep{fransson2006delaying}.

\begin{figure}
\centering
\includegraphics[width=0.55\textwidth]{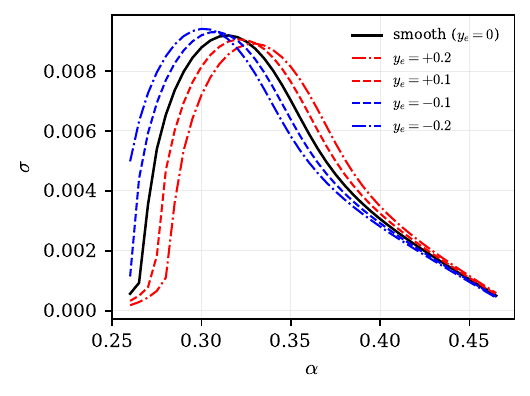}
\caption{Second-mode growth rates over shallow, dense, non-separating
roughness modelled by the reactive admittance \eqref{eq:Ar}.  The sign of
the effect reverses near $\alpha\approx0.33$, where $\arg K$ crosses
$-\pi/2$ (figure~\ref{fig:5}b), as Corollary~\ref{cor:phase} requires.}
\label{fig:4}
\end{figure}

\subsection{The distinguished phase band}
\label{ssec:band}

Figure~\ref{fig:5}(b) shows $K(\alpha)$ along both rigid branches.  At the
second-mode growth maximum,
\begin{equation}
K=-0.00238-0.00750\,\ii,\qquad |K|=0.00787,\qquad
\varphi_K=-107.6^\circ ;
\label{eq:K0}
\end{equation}
along the branch $\varphi_K$ rotates from $\approx-152^\circ$ in the
synchronisation region through $-108^\circ$ at the peak to $+49^\circ$ at
the upper edge.  The stabilising half-plane of Corollary~\ref{cor:phase}
rotates with it.  Phases near the Stokes value $-45^\circ$ and the coating
value $-39^\circ$ remain stabilising over essentially the whole band ---
whence the robustness of the damping in figures~\ref{fig:2} and
\ref{fig:3} --- pure resistance ($\varphi=0$) is stabilising for
$\alpha\lesssim0.37$, and the reactances $\pm90^\circ$ change sides exactly
where $\varphi_K$ crosses $\mp90^\circ$, as figure~\ref{fig:4} confirms.
For mode~I at low $\alpha$ the phase differs by nearly $180^\circ$,
reversing these r\^oles.  This \emph{distinguished band} of stabilising
phases, an attribute of the rigid-wall mode alone, is the practical content
of Theorem~\ref{thm:sens}: one complex number per mode and wavenumber
determines the first-order effect of every thin wall mechanism, alone or in
combination.

\section{Discussion}
\label{sec:discussion}

The contribution of this paper is one theorem and one mechanics
demonstration.  Theorem~\ref{thm:sens} is an explicit spectral perturbation
result for the compressible Rayleigh operator under boundary perturbation:
the first-order response of a simple trapped eigenvalue to \emph{any} wall
admittance is $KA$, with $K$ a closed-form functional of the rigid-wall
eigenfunction, and with the phase criterion of Corollary~\ref{cor:phase} as
its immediate consequence.  Section~\ref{sec:mechanisms} shows that the
admittance is precisely the object that thin-wall asymptotics delivers ---
for viscous sublayers, blind-pore coatings and shallow non-separating
roughness, under the explicit ordering of Assumption~\ref{ass:thin}, with
leading-order additivity --- and \S\ref{sec:numerics} validates both the
coefficient and its phase structure (the
sign-reversal of the roughness effect at $\arg K=-\pi/2$) at Mach 4.5.
The framework gives a leading-order spectral and asymptotic description of
how thin wall mechanisms modify trapped compressible boundary-layer modes;
it is intended as a diagnostic and design calculus complementary to, not a
replacement for, full viscous stability theory, direct simulation or
spatial transition prediction.

Limitations delimit the claims.  The analysis is locally parallel and
two-dimensional; oblique and non-parallel extensions require only the
corresponding rigid-wall eigenfunctions in \eqref{eq:thm}, the admittances
being unchanged.  Assumption~\ref{ass:spectral} excludes nearly degenerate
eigenvalues: in synchronisation regions the linear range shrinks and a
two-mode (degenerate) perturbation theory would be required.
Assumption~\ref{ass:thin} excludes isolated or high roughness, separated
flow over textures, and coating regimes with strong pore--pore or
pore--mode interaction; resonant reinforcement and coating-induced modes
documented at other conditions \citep{bres2013second,fedorov2011transition} lie outside
the attenuative regime computed here.  The quantitative results retain the
full admittance boundary condition; the first-order formula is exact only
in the limit, and at $\mathcal O(1)$ admittance it is used as a phase diagnostic
(Remark~\ref{rem:O1}). A different perturbation problem arises near synchronisation, where two
rigid-wall eigenvalues are separated by a gap comparable with the admittance
perturbation.  Then the scalar expansion \(c(A)=c_0+KA+\mathcal O(|A|^2)\) must be
replaced by a finite-dimensional perturbation on the nearly resonant
eigenspace.  Formally, if two rigid-wall modes \(c_1,c_2\) are retained, the
leading approximation has the form
\[
        \det\left[
        \begin{pmatrix}
        c_1&0\\[2pt]
        0&c_2
        \end{pmatrix}
        +
        A
        \begin{pmatrix}
        K_{11}&K_{12}\\[2pt]
        K_{21}&K_{22}
        \end{pmatrix}
        -
        cI
        \right]=0,
\]
where the diagonal entries reduce to the scalar sensitivities away from the
interaction, while the off-diagonal entries are adjoint-mode boundary
bilinear forms coupling the two branches.  Developing and validating this
two-mode theory is beyond the scope of the present paper; here we restrict
Theorem~\ref{thm:sens} and its validation to isolated simple eigenvalues.

\appendix

\section{Extensions to amplitude-dependent and acoustic admittances}
\label{appsec:extensions}

Two extensions illustrate the reach of the boundary condition
\eqref{eq:wallBC} beyond linear, trapped dynamics; both are presented as
outlines, not as developed theories.

\emph{Amplitude dependence.}  Oscillatory flow through pore orifices
acquires, at finite amplitude, the quadratic resistance of nonlinear liner
acoustics \citep{ingard1967acoustic,melling1973acoustic,cummings1986transient}.  In the
describing-function (first-harmonic) approximation the coating impedance
becomes $Z_{\rm NL}=A_p^{-1}+\Theta|\hat v_w|$ with
$\Theta=(1-\phi^2)/(2C_D^2\phi^2T_w)$, yielding a scalar fixed point for
$A_{\rm NL}(\varepsilon)$ at wall-pressure amplitude $\varepsilon$, and a
quasi-linear growth rate $\sigma(\varepsilon)$ from the same eigenvalue
problem.  Since $A_{\rm NL}\to0$ as $\varepsilon\to\infty$, the coating
benefit erodes universally with the scaled amplitude
$x=\Theta|A_p|\varepsilon$, like $x^{-1/2}$ at large $x$; computations at
$\alpha=0.31$ (figure~\ref{fig:6}) give half-effectiveness amplitudes
$\varepsilon_{1/2}\approx3$, far above the linear-disturbance regime, so
the reference coating is amplitude-robust over the range in which linear
theory itself applies.  For configurations that are linearly stabilised
($\sigma(0)<0<\sigma_{\rm rigid}$) the construction yields a
\emph{describing-function erosion threshold} $\varepsilon^\ast$ with
$\sigma(\varepsilon^\ast)=0$ (figure~\ref{fig:6}a).  This threshold is not
a full nonlinear transition threshold --- no amplitude equation for the
outer wave is derived --- it is the amplitude at which the
amplitude-dependent wall admittance no longer offsets the linear rigid-wall
growth rate.

\emph{Acoustic admittances.}  The far-field exponent \eqref{eq:farfield}
bounds the trapped regime; if wall or flow modifications push the phase
speed across $c_r<1-1/M$, the eigenvalue continues onto a leaky sheet and
the mode radiates a Mach wave at
$\theta_M=\arccos\{1/[M(1-c_r)]\}$, the wall admittance entering the leaky
dispersion relation through the same condition \eqref{eq:wallBC}
\citep{chuvakhov2016spontaneous}.  For the trapped modes computed here, sound is
generated where wall homogeneity is broken (coating junctions); 
in the Born approximation the receptivity of a junction with admittance jump
$\Delta A$ to sound incident at angle $\theta$ is proportional to the
wall-response function $W(\theta)=\hat p(0)/\hat p_{\rm inc}$ of the
signalling problem --- the same Rayleigh operator and wall condition with
incident and outgoing far-field branches --- weighted by the junction
Fourier transform, and flow-acoustic reciprocity implies the same kernel
controls the emission directivity. 
We do not develop the adjoint machinery
required to make the emission statement rigorous --- for the receptivity
half this is the programme carried through asymptotically by
\citet{dong2020receptivity}, see also \citet{fedorov2001prehistory,ma2003receptivity} --- and accordingly
present $W(\theta)$ as an acoustic diagnostic of wall treatments.  Computed
at the second-mode peak frequency (figure~\ref{fig:7}), the fast-wave
response is broad, peaking at $|W|=2.44$ near $\theta=43^\circ$ for the
rigid wall and reduced at \emph{all} angles by the coating (peak $-9\%$):
the admittance that damps the trapped mode also softens the wall response
to fast incident sound.  Slow waves are critical-layer-shielded
($|W|\lesssim0.2$ away from the grazing degeneracy).  A time-domain
impedance formulation \citep{tam1996time} would be required to carry these
admittances into direct simulation.

\begin{figure}
\centering
\includegraphics[width=\textwidth]{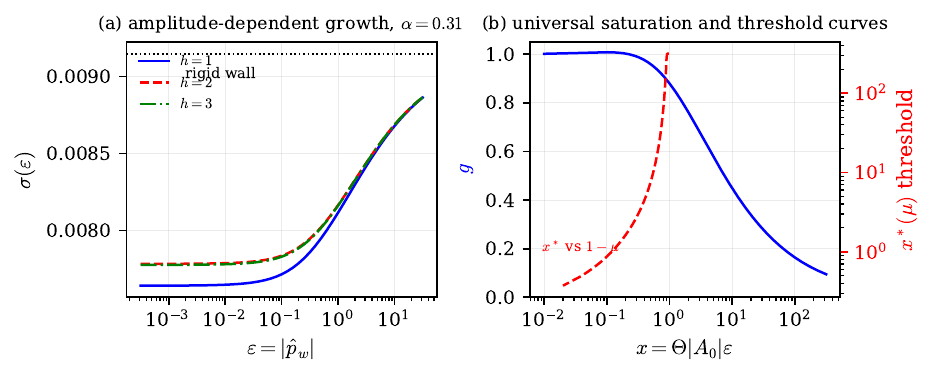}
\caption{Extension: amplitude-dependent coating.  (a) Quasi-linear growth
rate $\sigma(\varepsilon)$ at $\alpha=0.31$ for $h=1,2,3$, eroding toward
the rigid value (dotted).  (b) Universal saturation curve $g(x)$ (left
axis) and describing-function erosion threshold $x^\ast(\mu)$ (right axis,
against $1-\mu$, with $\mu$ the ratio of rigid growth to linear coating
stabilisation).}
\label{fig:6}
\end{figure}

\begin{figure}
\centering
\includegraphics[width=\textwidth]{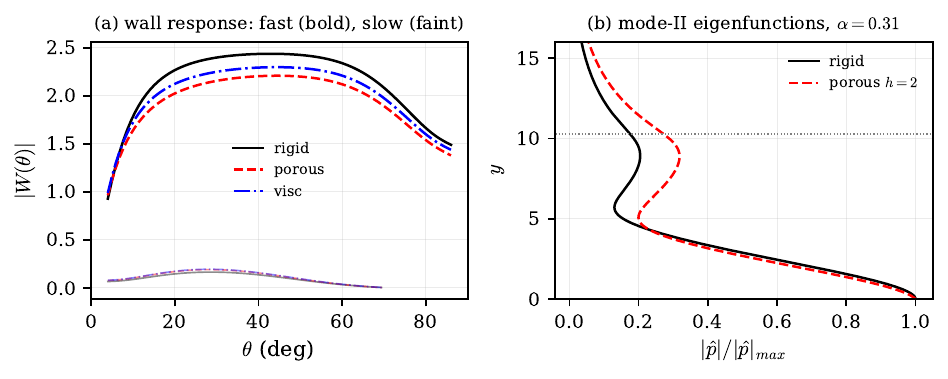}
\caption{Extension: acoustic diagnostic.  (a) Wall-response function
$|W(\theta)|$ at the second-mode peak frequency $\omega=0.2815$ for fast
(bold) and slow (faint, $\theta\le70^\circ$) incident waves, over rigid,
porous and viscous walls.  (b) Second-mode pressure eigenfunctions at
$\alpha=0.31$ over rigid and coated walls; dotted line: $\delta_{99}$.}
\label{fig:7}
\end{figure}

\section{The shear-affected wall layer: Airy form and the triple-deck
connection}
\label{app:airy}

When $\omega\not\gg\alpha\lambda_w\delta_v$, the base shear
$U\simeq\lambda_wy$ must be retained in the wall layer.  With
$\eta=y/\delta_A$, $\delta_A=(\nu_w/(\alpha\lambda_w))^{1/3}$, the shear
perturbation $\tau=\dd\hat u/\dd\eta$ obeys Airy's equation
\[
\frac{\dd^2\tau}{\dd\eta^2}-\ii(\eta-\eta_0)\tau=0,
\qquad
\eta_0=\frac{c}{\lambda_w\delta_A},
\]
with decaying solution
$\tau\propto\Ai(\ee^{\ii\pi/6}(\eta-\eta_0))$.  Imposing no slip, matching
the outer slip velocity $\hat u_s=T_w\hat p_w/c$, and integrating
continuity across the layer gives the uniformly valid admittance
\begin{equation}
A_v^{\rm Airy}
=\frac{\ii\alpha T_w\delta_A}{c}\,
\frac{\int_{\xi_0}^\infty\Ai(s)\,\dd s}{\ee^{\ii\pi/6}\Ai'(\xi_0)},
\qquad
\xi_0=-\ee^{\ii\pi/6}\eta_0,
\label{eq:AvAiry}
\end{equation}
the Tietjens-function form of asymptotic stability theory.  As
$|\eta_0|\to\infty$, $\Ai'(\xi_0)/\!\int_{\xi_0}^\infty\Ai\to
(-\xi_0)^{1/2}$ collapses \eqref{eq:AvAiry} onto the Stokes result
\eqref{eq:Av}; in the opposite limit $\delta_A$ becomes the lower-deck
thickness and \eqref{eq:AvAiry} reproduces the lower-deck displacement
response of triple-deck theory \citep{smith1979non} --- the
Tollmien--Schlichting wave is, in this language, an eigenoscillation of the
inviscid upper flow closed by the wall admittance \eqref{eq:AvAiry}, the
impedance interpretation advocated for receptivity by \citet{dong2020receptivity}.
For the second mode at $\Rey=2000$,
$|\eta_0|\gg1$ and \eqref{eq:Av}--\eqref{eq:AT} are accurate to a few per
cent.

\section{Pore-scale derivation of the Zwikker--Kosten coefficients}
\label{app:pore}

In a pore ($r<r_p$, $-h<z<0$) driven at its mouth by
$\hat p_w\ee^{-\ii\omega t}$, with $r_p$ small compared with the pore
acoustic wavelength and the external scales, the pressure is uniform over
the cross-section and the axial momentum and energy equations reduce to
\[
-\ii\omega\bar\rho_w\hat u=-\hat p_z+\mu_w\frac1r(r\hat u_r)_r,
\qquad
-\ii\omega\bar\rho_wc_p\hat\theta=-\ii\omega\hat p
+\frac{\mu_wc_p}{Pr}\frac1r(r\hat\theta_r)_r,
\]
with $\hat u=\hat\theta=0$ at $r=r_p$.  Both are Womersley problems; with
$k=(\ii\omega/\nu_w)^{1/2}$ the cross-sectional averages are
\[
\langle\hat u\rangle=\frac{\hat p_z}{\ii\omega\bar\rho_w}[1-H(kr_p)],
\qquad
\langle\hat\theta\rangle=(\gamma-1)M^2T_w\,\hat p\,[1-H(kr_p\sqrt{Pr})],
\]
with $H(z)=2J_1(z)/[zJ_0(z)]$.  The first defines $\rho_d$ of
\eqref{eq:ZK}; substituting the second into the linearised state and
continuity equations gives $C_d$.  Averaged continuity and momentum form
the transmission line $\hat p_z=\ii\omega\rho_d\langle\hat u\rangle$,
$\langle\hat u\rangle_z=\ii\omega C_d\hat p$, whose solution with
$\langle\hat u\rangle(-h)=0$ yields the input admittance
$-\ii Y_0\tan(\Lambda h)$ and, after multiplication by the porosity,
\eqref{eq:Ap}.  End corrections and interaction with the overlying Stokes
layer are $\mathcal O(r_p/\delta_v,\phi A_v)$, consistent with
Assumption~\ref{ass:thin}.

\section{The far-field tail contribution $I_\infty$}
\label{app:tail}

On $y>y_b$ both eigenfunctions take the form \eqref{eq:farfield} with
$\gamma=\gamma(c)$.  Substituting
$\hat p_0=P_b\ee^{-\gamma_0(y-y_b)}$ and the perturbed exponent
$\gamma_0+\delta c\,\dd\gamma/\dd c$ into the upper boundary term of
\eqref{eq:identity}, with $\Phi_0\to1/(1-c_0)^2$,
$\dd\gamma/\dd c=\alpha^2M^2(1-c_0)/\gamma_0$, gives
\[
-\Phi_0(y_b)\big(\hat p_0\,\delta\hat p'-\hat p_0'\,\delta\hat p\big)_{y_b}
=\delta c\,\frac{P_b^2}{(1-c_0)^2}\frac{\dd\gamma}{\dd c}
=\delta c\;\frac{2}{(1-c_0)^3}\,
\frac{(\gamma_0^2+\alpha^2)P_b^2}{2\gamma_0}\;=\;\delta c\,I_\infty,
\]
where $\alpha^2M^2=(\gamma_0^2+\alpha^2)/(1-c_0)^2$ by
\eqref{eq:farfield}.  Equivalently, $I_\infty$ equals the volume integral
of \eqref{eq:Idef} evaluated in closed form over the uniform stream
$(y_b,\infty)$ with $U=T=1$ and the tail eigenfunction; the two derivations
agree identically, confirming that the truncation at $y_b$ introduces no
error in $K$.

\section{Numerical details and convergence}
\label{app:num}

The base flow solves $f'''+ff''=0$, $f''(0)=0.46960$ in the
Dorodnitsyn--Howarth variable, with
$T=1+\tfrac12(\gamma-1)M^2(1-U^2)$ and $y=\int_0^\eta T\,\dd\eta$,
tabulated on $4001$ points to $\eta=14$ and spline-interpolated.  The
Rayleigh march uses fixed-step RK4 from $y_b=20.83$ on $4000$ intervals
with half-grid coefficients; eigenvalues agree to nine significant figures
with adaptive eighth-order integration at tolerance $10^{-12}$ and are
unchanged to all quoted figures under doubling of $y_b$ or the step count.
Newton iteration (numerical derivative, step $10^{-7}$, damped to
$|\delta c|\le0.05$) converges in $4$--$8$ iterations; $c$-dependent
admittances converge in $2$--$4$ outer sweeps.  The signalling problem of
\S\ref{appsec:extensions} reuses the marcher from the wall outward,
decomposing the solution at $y_b$ into incoming/outgoing exponentials
selected by the vertical group velocity in the moving stream; the slow-wave
critical layer is regularised by $c\to c+0.003\ii$, with results
insensitive to halving the regularisation except within $1^\circ$ of the
grazing degeneracy.


\end{document}